\documentclass{amsart}
\usepackage{ulem}
\usepackage{amsfonts}
\usepackage{amsmath}
\usepackage{amssymb}
\usepackage{geometry}
\usepackage{mathtools}
\usepackage{color}
\usepackage{amsthm}
\usepackage{ifpdf}
\usepackage{todonotes}
\usepackage{graphicx, caption}
\usepackage{bm}
\usepackage{comment}
\usepackage{pdflscape,afterpage}
\usepackage[pdftex,
  	bookmarks=true,
	unicode=true,
	pdfborder={0 0 0},
	pdfstartview={FitV},
	pdftitle={},
	colorlinks=false,
	linktoc=all,
	plainpages=false,
	pdfpagelabels,
	pdffitwindow=true]{hyperref}

\makeatletter
\@namedef{subjclassname@2020}{2020 MSC}
\makeatother

\ifpdf
\else
\fi
\theoremstyle{plain}
\newtheorem{theorem}{Theorem}[section]
\newtheorem{corollary}[theorem]{Corollary}
\newtheorem{lemma}[theorem]{Lemma}
\newtheorem{proposition}[theorem]{Proposition}
\theoremstyle{definition}

\newtheorem{example}[theorem]{Example}
\theoremstyle{remark}
\newtheorem{remark}[theorem]{Remark}
\newcommand{\R}{\mathbb{R}}

\renewcommand{\bar}{\overline}
\renewcommand{\tilde}{\widetilde}
\renewcommand{\hat}{\widehat}

\numberwithin{equation}{section}
\input{tcilatex}

\def\bh{\bar{h}}
\def\eps{\varepsilon}
\def\I{\Lambda}
\def\J{{\mathcal{J}}}
\def\s{\sigma}
\def\W{\mathrm W}
\def\A{\mathrm A}
\def\V{\mathrm V}
\def\v{\mathrm v}
\def\X{\mathrm X}
\def\WW{\mathbf W}
\def\VV{\mathbf V}
\def\vv{\mathbf v}
\def\h{\mathrm h}
\def\k{\mathrm k}
\def\l{\mathrm l}
\def\H{\mathrm H}
\def\C{\mathcal C}
\def\hh{\mathbf h}
\def\kk{\mathbf k}

\def\beps{\bar \eps}
\def\RR{\mathbb{R}}
\def\MM{\mathbf M}
\def\M{\mathcal M}
\def\TH{{\rm H_0}}
\def\CM{{\rm H}}

\begin{document}
\title[Rough vol asymptotics]{Precise asymptotics: robust stochastic volatility models}
\author{ P. K. Friz, P. Gassiat, P. Pigato}
\address{TU and WIAS Berlin, U\ Paris Dauphine, WIAS Berlin}
\date{\today }

\begin{abstract}
We present a new methodology to analyze large classes of (classical and rough) stochastic volatility models, with special regard to short-time and small noise formulae for option prices. Our main tool is the theory of regularity structures, which we use in the form of [Bayer et al.; A regularity structure for rough volatility, Math. Finance 2019].  In essence, we implement a Laplace method on the space of models (in the sense of Hairer), which generalizes classical works of Azencott and Ben Arous on path space and then Aida, Inahama--Kawabi on rough path space. When applied to rough volatility models, e.g. in the setting of [Bayer et al.; Pricing under rough volatility, Quant. Finance 2016] and [Forde-Zhang, Asymptotics for rough stochastic volatility models, SIAM J. Financial Math. 2017], one obtains precise asymptotics for European options which refine known large deviation asymptotics. 
\end{abstract}

\thanks{We gratefully acknowledge financial support of European Research Council Grant CoG-683164 (PKF
and PP) and ANR-16-CE40-0020-01 (PG). We are also grateful to the anonymous reviewer for his/her careful reading.}

\maketitle
\tableofcontents

\section{Introduction}
Consider a generic stochastic volatility model, written in the form
$$
dS_t/S_t = \sigma (t,\omega) d\tilde W _t + \mu (t,\omega)dt \;.
$$
This includes the Black--Scholes model with constant $\sigma, \mu$, and many ``classical'' (Markovian) stochastic volatility (short: StochVol) models such as the Heston model \cite{heston1993closed}, the Stein--Stein model \cite{stein1991stock} or the SABR model \cite{hagan2002managing, hagan2015probability}. The Brownian motion $\tilde W$ is typically decomposed in $\rho W + \bar{\rho} \bar{W}$ such that $\sigma$ is a $W$-adapted diffusion in its own right. Following e.g. \cite{berestycki2004computing, deuschel2014marginalb}, by a general (classical) StochVol model we mean an $n$-dimensional diffusion, whose first component has the interpretation of (log)price.

\medskip 

In contrast, the recent class of rough 
volatility models (in short: RoughVol) assumes that $\sigma$ is an anomalous diffusion with scaling (Hurst) parameter $H < 1/2$. First introduced by Al\`os et al. \cite{alos2007short}, the interest 
in these models has exploded over recent years, starting with Gatheral et al. \cite{gatheral2018volatility} which not only presents statistical evidence but 
also hints to a market microstructure foundation (which subsequently led to the RoughHeston model \cite{el2019characteristic}). On the pricing side, one needs to fit these (parametric) models to given market data.  A well-documented problem, see e.g. \cite{fouque2000derivatives, gatheral2011volatility}, is to fit short-dated option prices (equivalently: volatility smiles), which in turn explains the interest of practioners in asymptotic formulae for short-dated options.  Some of these, like Hagan's SABR formula, have found widespread use in markets. 

\medskip 
Let us quickly review the case of {\bf classical StochVol}. Mathematically, the study of option prices 
(equivalently: Black--Scholes implied volatilities)  in diffusion based models in the short-time limit is intimately related to {\it large deviations}. To appreciate this fundamental link, the following heuristic
is useful. Take $K < S_0, \mu \equiv 0$ and consider the (out-of-the-money) {\it digital put} with price
\begin{equation}
            P [ S_t < K ] = P [ \log (S_t / S_0) < \log (K/S_0)] \approx e^{-\mathcal{I}(k)/ t} \; \text{ as $t \to 0$} \;.
\end{equation}
Here $k = \log (K/S_0)< 0 $ is called log-strike, $\mathcal{I}(k)$ energy (or rate) function. This is a typical (Freidlin--Wentzell) large deviation statement, valid as $t \to 0$, and $\mathcal{I}(k)$ has the geometric interpretation of shortest square-distance to some arrival manifold determined by $k$. It is not hard to check that put option prices have similar behaviour,
\begin{equation} \label{equ:PUT0}
         E [ (K - S_t)^+] \approx e^{-\mathcal{I}(k)/ t} \; \text{ as $t \to 0$} \;.
\end{equation}
In the Black-Scholes model, the (Gaussian) rate function is given by $k^2 / (2 \sigma^2)$; equating this expression with $\mathcal{I}(k)$, leads to the notion of effective (or Black--Scholes) implied volatility $\sigma_I (t,k)$ so that, as $t \to 0$, 
\begin{equation} \label{equ:IVexp}
                 \frac{k^2}{2 \sigma^2_I (t,k)} \sim  \mathcal{I}(k)   \text{ or equivalently:} \ \  \sigma^2_I (t,k) \sim \frac{k^2}{{2 \mathcal{I}(k)}} =: \Sigma_0(k) \;.
\end{equation}
This final relation is known as BBF formula  \cite{berestycki2004computing}; see also Pham \cite{pham2010large} for a derivation.
There is much interest from practitioners in refined asymptotics such as\footnote{This is no replacement for numerical methods, but allows to understand model parameters, helps to find good ``starting points'' for calibration procedures, and has led
to widely used parametrisations of the implied volatility surface.}
$$
\sigma^2_I (t,k) =  \Sigma_0(k) + t  \Sigma_1(k) + ... .
$$
which in turn requires a refinement of the (large deviation) aysmptotics (\ref{equ:PUT0}) to what may be called ``precise asymptotics''. We note that ``(any-order) call / put price asymptotics $\implies$(any order) implied volatiliy asymptotics'' amounts to an asymptotic understanding of the Black--Scholes formula and is the content of the decisive work of Gao--Lee \cite{gao2014asymptotics}. The mathematical problem is hence reduced to option price asymptotics.
In the Heston model, for instance, such an expansion was established in \cite{forde2012small}, using saddle point methods to obtain the necessary option price asymptotics from the known Heston characteristic function. In an absence of such knowledge, this approach fails for general (classical) StochVol models. Many authors have focused on density expansions instead, building on the vast existing literature on heat-kernels. For an overview, the practioners's textbook \cite{henry2008analysis} has numerous references, see also \cite{friz2015large} for a recent collection of works in this area. 


\bigskip

Let us turn to the case of {\bf RoughVol}. Following {\cite{alos2007short, fukasawa2011asymptotic, bayer2016pricing, fukasawa2017short, forde2017asymptotics} we consider the model case 
\begin{equation} \label{equ:simpleRV}
       \sigma (t,\omega) = \sigma ( W^H_t) 
\end{equation} 
where, in abusive notation, $\sigma$ also denotes a volatility function, taken of (at most) linear growth in \cite{forde2017asymptotics} or in exponential form in \cite{bayer2016pricing}, dubbed RoughBergomi.\footnote{A time-dependent volatility function would accomodate that case of non-constant forward variance curve, but in view of our focus on short-time, this extension adds little to our discussion.} 
Here $W^H$ is a fractional Brownian motion driven by a Brownian motion $W$, $\rho$-correlated with $\tilde{W}$ driving the price $S$.
Another popular instance in this family is given by RoughHeston, in which case $\sigma(t,\omega)$ is non-triviallly given as solution to a stochastic Volterra equation with singular kernel of the form $|t-s|^{H-1/2}$, so that $\sigma (t,\omega)$ again has $H^-$-H\"older sample paths. An appealing feature of RoughHeston is again information about its characteristic function \cite{el2019characteristic} which suggests a (finite-dimensional) asymptotic analysis approach to short-dated options under RoughHeston; this has been carried out in a recent preprint by Forde et al. \cite{forde2019small}. The main difference between StochVol and RoughVol is of course (local) scaling $H=1/2$ vs. $H < 1/2$. This already affects the basic large deviations: there is no LDP, as $t \to 0$ (equivalently: $\eps\to 0$) for
$$
                 \int_0^t \sigma ( W^H_s) d( \rho W_s + \bar{\rho} \bar{W}_s) \;;  \ \ \ \ \text{equivalently: } 
                 \ \int_0^1 \sigma ( \eps^{2H} W^H_s) \eps d( \rho W_s + \bar{\rho} \bar{W}_s) \;.
$$
Only if one fixes the scaling (in the form of an additonal factor $t^{H-1/2}$, equivalently: $\eps^{2H-1}$) are we in a small noise setting and can expect a large deviation principle with speed $t^{2H}$ or, equivalently, $\eps^{4H}$. 
This was first shown by \cite{forde2017asymptotics} and then, with no growth restriction on $\sigma$, in 
\cite{bayer2019regularity}, see also \cite{jacquier2018pathwise, gulisashvili2017large}.
Applied to the pricing of put options one then obtains \cite[Cor. 4.9]{forde2017asymptotics},
\begin{equation} \label{equ:PUT}
         E [ (K_t - S_t)^+] \approx e^{-\J(x)\,/ t^{2H}} \; \text{ as $t \downarrow 0$} \;
\end{equation}
with time-dependent strike $K_t = S_0 \exp ( xt^{1/2-H}) < S_0$ and energy function $\J$ specific to (\ref{equ:simpleRV}).
As was noted in \cite{forde2017asymptotics} such option price asymptotics imply a ``rough BBF formula'' of the form
\begin{equation} \label{equ:roughBFF}
                  \sigma_I^2 (t,xt^{1/2-H}) \sim \frac{x^2}{{2 \J(x)}} = \Sigma_0(x) \;.
\end{equation}

\bigskip 

\noindent {\bf Contributions of this paper.}

\smallskip 

\noindent ${\bf (I)}$ Our general main result provides a simultaneous refinement of (\ref{equ:PUT0}) and (\ref{equ:PUT}) as follows.  For a concise formulation we introduce normalized log-price\footnote{Financially minded readers might prefer $S_t / F_t$ where $F_t$ is the time-$t$ forward price, such as $S_0 e^{rt}$, but this adds very little to our discussion of short-time asymptotics.} 
$$
		X_t := \log (S_t / S_0)
$$
and (put, call) option prices, with log-strike $k$, as follows,
$$
                p(t,k) :=  E [ (e^k - e^{X_t})^+]  , \qquad c(t,k)  :=  E [ ( e^{X_t}- e^k)^+] \;.
$$

\begin{theorem}\label{thm:main000} (See Theorem \ref{thm:main0} for a precise formulation.) Consider a ``robust'' (classical or rough) stochastic volatility model with log-price process $(X_t)$ with spot-vol $\sigma_0>0$. Consider European put prices $p=p(t,k)$ with (out-of-money) log-strikes $k_\eps  = x  \eps^{1-2H} < 0$. Under a non-degeneracy assumption for the most-likely path, there exists a rate function $\I=\I(x)$, regular near $x$, and a function $A=A(x) \sim 1$ as $x \to 0$, such that, with $\sigma_x^2=2\I(x) / \I'(x)^2$, we have small noise asymptotic (out-of-money) put option prices as follows,
\begin{equation} \label{equ:main0000}
       p (\eps^2, k_\eps) 
       \sim \exp \left( {  - \frac{\I(x)}{\eps^{4H}}  } \right)  \eps^{1+4H} \, 
       \frac{A(x)} 
       { (\I'(x))^2  \sigma_x \sqrt{2\pi}}\ \ \ \text{as $\eps \downarrow0$} .
\end{equation}
Under a ``$1^+$ moment assumption'' a similar results holds for out-of-the money call option prices. 
\end{theorem}
\noindent Several remarks are in order. 

\smallskip

\noindent {\bf(i)}  Existence of a first moment for $\exp(X_t)$ is a ``conditio sine qua non'' for martingale based pricing theory. In particular, put-call parity holds which one can use to trivially deduce asymptotic formulae for in-the-money options. For the call expansion only, we assume existence of ``$1^+$ moments''. This is a rather mild condition but see also Remark \ref{rem:A2}.

\noindent  {\bf(ii)} We can switch from small-noise to short time by setting $\eps^2 \equiv t$. 

\noindent  {\bf(iii)} A (standard) application of Gao--Lee \cite{gao2014asymptotics}, detailed in Appendix D for the reader's convenience, shows that our price expansion implies a short-dated (squared) implied volatility approximation of the form
\begin{equation}\label{e:IVI}
\sigma^2_I (t,xt^{1/2-H}) =  \Sigma_0(x) + t^{2H}  \Sigma_1(x) + o(t^{2H}) \;,
\end{equation}
where $\Sigma_0$ is given in (\ref{equ:roughBFF}) and we have the explicit refinement
\begin{equation}\label{def:a}
\Sigma_1(x)
=
\begin{cases}
\frac{
x^2 }{
2\Lambda(x)^2
  }
 \log \bigg( \frac{ 2A(x)\Lambda(x) }{ \Lambda'(x) x }
\bigg) & \mbox { if } H<1/2 \\
\frac{
x^2 }{
2\Lambda(x)^2
  }
 \log \bigg( \frac{ 2A(x)\Lambda(x) }{ \Lambda'(x) x \exp(x/{2})}
\bigg) & \mbox { if } H = 1/2 \;.
\end{cases}
\end{equation}

\noindent  {\bf(iv)} In the generality of Theorem \ref{thm:main000} one cannot hope for explicit formulae of $\Lambda, A$ (and then $\Sigma_0,\Sigma_1$). As is typical for such rate functions, $\Lambda=\Lambda(x)$ can be written as variational problem, cf. (\ref{equ:inducedRF}).  On the other hand, we will see in the later proof, cf. (\ref{equ:Aabstract}),
\begin{equation} \label{equ:Aintro}
 A(x) = 
\begin{cases}
    E \big[ \exp ( \I'(x) \Delta_2^x ) \big],& \text{if } H < 1/2\\
    e^x E \big[ \exp ( \I'(x) \Delta_2^x ) \big],& \text{if } H=1/2
\end{cases}
\end{equation}
where $\Delta_2^x$ is a certain quadratic Wiener functional (that can also be further described). All this means that - in specific models - further computations are possible, notably expansions around $x=0$. This in turn allows to further quantify and put numerically to test \cite{friz2019explicit} the general higher order implied volatility expansion given in (\ref{e:IVI}).

\noindent  {\bf(v)} Many density (and then option price, implied volatility) expansions for generic classical StochVol models depend on existing heat-kernel expansions and thus
do not extend to the non-Markovian case of $H<1/2$. We recall that Malliavin calculus provides another route to heat-kernel expansions and this underlies the work of Kusuoka--Osajima \cite{kusuoka2008remark, osajima2015general} and also Deuschel et al. \cite{deuschel2014marginala, deuschel2014marginalb} which focus on density expansions -  option price expansions can be obtained upon suitable integration thereof. Our methods are different and allow us to tackle the option prices directly; Malliavin calculus, H\"ormander-type arguments and densities are not required. 




\smallskip 

\noindent {\bf(vi)}  Our ``robustness'' assumption'', in case $H=1/2$, means that the (log-)price process is the continuous image of the {\it It\^o Brownian rough path} and generically satisfied for SDE models with sufficiently smooth (say $C^{2^+}$) coefficients. (In a sense, this supersedes the notation of ``regular Wiener functional'' central to \cite{kusuoka2008remark}.) Despite some recalls in Section \ref{sec:SDEs} below, we assume some familiarity with basic ideas of rough paths, as found e.g. in \cite{friz2014course}. Interestingly, when
$H<1/2$, (standard) rough paths are ill-suited for RoughVol models\footnote{When $H<1/2$, the integrand $\sigma(W^H)$ is not controlled by the Brownian
 noise.} but the theory of regularity structures \cite{hairer2014theory} comes as a perfect substitute, as was seen in \cite{bayer2019regularity}.
 The assumption now means that the price process can be constructed as continuous image of a suitable {\it It\^o model}, where (in Hairer's
 terminology) ``models'' generalize rough paths. For the reader's convenience, the construction is reviewed in the Appendix. 

\smallskip 

\noindent {\bf(vii)} (Aida, Inahama-Kawabi) In the SDE context, the rough path approach is known to offer many benefits, most prominently continuous dependence on the noise (viewed as random rough path), which in turn is of immediate interest in establishing Freidlin--Wentzell type large deviations, see e.g. \cite[Ch.10]{friz2014course}.  In his work on semiclassical limits \cite{aida2007semi}, see also \cite{inahama2007asymptotic, inahama2008laplace, inahama2010stochastic}, Aida noted that rough path spaces are also most convenient to establish {\it precise Laplace asymptotics} for Wiener functional of the form $ E(\exp F(X^\eps) / \eps^2) \sim \text{(const)} \exp (-\frac{\Lambda_F(\gamma)}{\eps^2}), $
previously considered (in a classical path space setting) by Ben Arous \cite{arous1988methods} and others. Our theorem deals with different Wiener functionals, and produces different expansions, even in case $H=1/2$, and we rely on a non-trivial variation of the Laplace method, cf.  {(viii)} below. On a conceptual level, cf. {(vi)}, we replace rough paths with models. 

\smallskip 

\noindent {\bf(viii)} (Azencott) Perhaps closest to our main result, albeit in a classical diffusion ($H=1/2$) setting, Azencott \cite{azencott1985petites} gives asymptotic expansions for probabilities such as $P(X_1^\eps >k)$ using an intricate variation of the Laplace method on path space.  As Azencott himself laments, his approach is (very) technical but expresses hope that this may change: ``{\it Il serait tr\`es int\'eressant d'avoir une 
 justification syst\'ematique g\'en\'erale de la validit\'e des d\'eveloppements formels de ce type  et nous avons (\`a moyen terme!) une vue assez optimiste sur l'existence d'un formalisme indolore et garanti math\'ematiquement.}''.  Our proof of Theorem \ref{thm:main000}, if restricted to a classical diffusion setting, and adapted to ``digital payoffs'', essentially shows that rough paths, and for $H<1/2$: regularity structures,  are the tool that Azencott was hoping for. 

\smallskip 

\noindent {\bf(ix)}  (Heston) The ``robustness'' assumption essentially requires some smoothness of the coefficients in the model, which seems to rule out (classical and rough) Heston-type model because of square-root coefficients. While such degeneracies are not (at all) the focus of this work, we point out that that our expansion is determined by a neighbourhood of the most-likely path, in uniform (and even stronger) metrics. Hence, any ``initial'' localization to a uniform neighbourhood of the most-likely path, obtained by pathwise large deviations (such as \cite{robertson2010sample} in the Heston case), essentially allows to ignore the square-root issues and to apply our theorem, consistent with known Heston results \cite[Thm 3.1]{forde2012small}, and more recently \cite{forde2019small}. That said, the available (affine) structure of both Heston and RoughHeston (characteristic function methods ...) leaves little motivation to work out the details of this remark.

\noindent {\bf(x)} (Local volatility) The asymptotic behavior of local volatility under RoughVol can be analyzed by combining the methods of this paper with Malliavin calculus (integration by parts), details are left to a forthcoming note. This allows to generalize the results of De Marco and one of us \cite{de2018local} to the rough volatility setting. 

\smallskip 

\noindent ${\bf (II)}$ {\bf Moderate Deviation Pricing.}  Replace $x$  in Theorem \ref{thm:main000} by $x \eps^{2\beta}$ and consequently consider out-of-the-money log-strikes $$k_\eps = x \eps^{1+2(\beta-H)}.$$ Since $\I (x \eps^{2\beta}) \sim $(const)$x^2 \eps^{4\beta}$ we see that exponential concentration is lost in (\ref{equ:main0000}) when $\beta \ge H$. The case $\beta = H$ is precisely the central limit regime in this context. Of course, the case $\beta=0$ is the large deviation case treated in ${\bf (II)}$. Inbetween, when  $\beta \in (0,H)$, we have the moderate deviation regime. Classic moderate deviations are a kind of refinement of the CLT and often not very informative, higher order moderate deviation on the other hand can contain interesting information as was seen in a classical StochVol (resp. RoughVol) setting in \cite{friz2017option} and \cite{bayer2017short}. The following ``precise'' moderate deviation result refines simultaneously the respective results in {\cite{friz2017option, bayer2017short},
\begin{equation} \label{equ:main000mod}
       p(\eps^2,k_\eps) \sim \exp \left( {  - \frac{\I(x)}{\eps^{4(H-\beta)}}  } \right)  \eps^{1+4(H-\beta)} \, 
       \frac{\sigma_0^3} 
       { x^2 \sqrt{2\pi}}\ \ \ \text{as $\eps \downarrow0$} .
\end{equation}
We do not spell out a similar call price formula. 
Moderate deviations in fact offer an appealing way to understand the rough BBF formula (\ref{equ:roughBFF}): it suffices to ``compare'' moderately out-of-the-money options in Black--Scholes (struck at $k_\eps = x \eps^{2\beta}$) with the honest (large deviation) out-of-the-money price of an option under RoughVol. Provided that $\eps^{4(1/2 - \beta)}$ matches $\eps^{4H}$, so that we must have $\beta \equiv 1/2 -H$, we can again deduce the leading order effective behaviour (\ref{equ:roughBFF}) of the Black-Scholes (implied) volatility (struck at $x \eps^{2\beta} = x \eps^{1-2H}$) as $\eps$ (equivalently: $t$) goes to zero. 
 
\bigskip 

\noindent ${\bf (III)}$ {\bf RoughVol}. The conditions of our main theorem can be verified for the RoughVol model  (\ref{equ:simpleRV}), including the case of lognormal volatility as seen in RoughBergomi \cite{forde2017asymptotics, bayer2016pricing}.  

\begin{corollary}[RoughVol] \label{thm:main} (cf. Corollary \ref{thm:main44} in Section 7.) Let $H \in (0,1/2]$ and $k_\eps = x  \eps^{1-2H} < 0$. 
Then, for $x$ small enough, $\J = \J(x)$ is continuously differentiable and we have (out-of-the-money) put prices with
\begin{equation}
       p(\eps^2,k_\eps) \sim \exp \left( {  - \frac{\J(x)}{\eps^{4H}}  } \right)  \varepsilon^{1+4H} \frac{A(x)}{ (\J'(x))^2  \sigma_x \sqrt{2\pi}} \ \ \text{as $\eps \downarrow0$,}
\end{equation} 
for some function $A(x)$ with $A(x) \to 1$ as $x \to 0$. Under a $1^+$ moment assumption, a similar expansion holds for out-the-money call prices ($x>0$).
\end{corollary}

\noindent We again make several remarks. 

\smallskip 

\noindent  {\bf(i)} The robustness condition is guaranteed by stochastic Taylor expansions, viewed through the lenses of regularity structures. This is the main point of \cite{bayer2019regularity}, reviewed in Section \ref{sec:roughDEs} and Appendix \ref{app:ERS} for the reader's convenience.

\noindent  {\bf(ii)} In Corollary \ref{thm:main} we verify the relevant abstract conditions, notably the non-degeneracy condition of Theorem \ref{thm:main000} for $x$ in a neighbourhood of zero. This is a subtle point since indeed we do not expect this for all $x$, for the same reason we must not expect the cut-locus $\mathrm{Cut}(x_0,x)$ to be empty for two generic points on a general Riemannian manifold. Readers familiar with heat-kernel expansions (valid away from the cut-locus) and especially the (Laplace based) proof given in \cite{ben1988developpement}, and/or the adaption to marginal density expansions \cite{deuschel2014marginala}, will realize that this is not just a loose analogy but in fact has an almost identical infinite-dimensional explanation; cf. the setting of Appendix \ref{app:SOO}. 


\noindent  {\bf(iii)} As mentioned in point (I.iii) above, the available expression for $\Lambda$ and $A$ are amenable to further computations once the model is further specified. The RoughVol model is such a specification, full described by $H$, the volatility function $\sigma(\cdot)$ and the correlation parameter $\rho$. In particular, these quantities 
determine $\Delta_2^x$ which in turn feed into the general formula (\ref{equ:Aintro}). An expression for $\Delta_2^x$, in case of RoughVol, is straight-forward from our general theory and is given in (\ref{equ:DeltaForRV}). This explicit form is the starting point in \cite{friz2019explicit}, where we further detail specific computations to RoughVol, and especially RoughBergomi, and put to numerical test. The resulting implied volatility expansions are seen therein to be very accurate, which is easily attributed this to the additional term $t^{2H} \Sigma_1(x)$ in the expansion, with $H$ between $0.1$ and $0.3$ in typical calibrations. 

\noindent  {\bf(iv)} Our work relies fundamentally on (precise) large deviations which leads to the validity of our expansions in the large (and also moderate) deviation regime. One can also study such expansions in the central limit regime (think: $k_\eps = x \eps$, equivalently $k_t = x \sqrt{t}$) and employ Edgeworth expansions to obtain refined information, leading also to order $t^{2H}$ correction terms; see \cite{euch2019short} and the references therein for this CLT based approach. We note that the large deviation range (think: $k_\eps = x \eps^{1-2H}$) can be significantly larger than the CLT range, although the gap narrows as $H \downarrow 0$. 

\noindent  {\bf(v)} Our work has nothing to say about the super-critical regime $H=0$. First steps in this direction, based on Kahane's theory of Gaussian multiplicative chaos, were taken in \cite{neuman2018fractional}.


%

\smallskip 

\noindent {\bf Organization of Paper.} Sections 2 and 3 contain notations and some (non-asymptotic) Black--Scholes estimates useful later on. Section 4 formalizes the basic large deviations assumption and gives a number of classical and ``rough'' examples where it is satisfied. Put and call price large deviations 
are given under natural conditions. Section 5 presents our (two) running examples, based on which the remaining abstract conditions are formulated. 
Section 6 contains the statements and partial proof our main results: exact option price formulae, precise large deviations, precise moderate deviations. The technical assumptions for the RoughVol case are verified in Section 7. Some long proofs have been moved to Section 8. Appendix A details the notion of ``robust representation", Appendix B recalls the rough volatility regularity structure, following \cite{bayer2019regularity}, Appendix C deals with control theoretic issues necessary for the local analysis around the minimizer, Appendix D extracts implied volatility - from option price expansions.

\section{Notation}

Wiener space $C([0,1],\R^m)$ supports $m$-dimensional standard Brownian motion $\W = \W(t,\omega)$ and contains the Cameron--Martin space $\H \equiv \H^1 \subset C([0,1],\R^m)$ of paths with square-integrable derivative. It is tacitly assumed that all Brownian and Cameron--Martin paths start at the origin. When $m=2$, write $\W = (W, \bar W)$ and $\h = (h, \bar h)$ accordingly; given a correlation parameter $\rho \in [1,1]$, and $\rho^2 + \bar \rho^2 = 1$, a scalar Brownian motion ($\rho$-correlated with $W$) is given by $\tilde W = \rho W + \bar \rho \bar W$. Similarly, write $\tilde h = \rho h + \bar \rho \bar h$. By Hurst parameter we mean a scalar $H \in (0,1/2]$. Convolution with the singular kernel $K^H (s,t) =  \sqrt{2H} |t-s|^{H-1/2}1_{0<s<t}$ is known to $\beta$-regularize, with $\beta=1/2+H$ (e.g. \cite[section 3.1]{SKM}). (When $H=1/2$, this is precisely indefinite integration.) Applied to scalar white noise $\dot W$, a.s. in the negative H\"older space $C^{-1/2^-}$, we get a Riemann--Liouville (a.ka. Volterra or L\'evy) fractional Brownian motion $\hat W = K^H * \dot W \in C^{H^-}$ a.s., where we follow standard notation in the field and write H\"older exponent $\alpha^-$, when we really mean $\alpha - \kappa$ for all sufficiently small $\kappa>0$. 
We also write $\hat h = K^H * \dot h \in H^{\beta}$, where $\dot h \in L^2 \equiv H^0$ and $H^\beta$ denotes fractional Sobolev space. Throughout, $\X = (X,Y)$ is a $n$-dimensional stochastic process, with scalar first component $X=X(t,\omega)$, defined on $m$-dimensional Wiener space.
We shall assume a robust form, which allows to write $X$ as continuous function of $\WW=\WW(\omega)$, which is a suitable enhancement of $\W$ to a random element in rough path space $\C^\alpha, \alpha \in (1/3,1/2)$ or more generally, a random element in the space of models $\M=\M^\kappa$, defined on a regularity structure whose precise form depends on the dynamics at hand. (The case of rough volatility is reviewed in detail in the Appendix.)

\section{Preliminaries on Black-Scholes asymptotics}

The Black-Scholes log-price at time $t=\varepsilon ^{2}$ is given by%
\[
X_{1}^{\varepsilon }=\varepsilon ^{2}\mu +\varepsilon \sigma B_{1}\sim
N\left( t\mu ,\sigma ^{2}t\right) \ .
\]%
Consider log-strike $k>0$ and set%
\[
Z_{1}^{\varepsilon }=\varepsilon ^{2}\mu +\varepsilon \sigma
(B_{1}+k/\varepsilon \sigma )=k+\varepsilon \sigma B_{1}+\varepsilon ^{2}\mu
.
\]
Then the call price is given by 
\begin{eqnarray*}
E\left[ \left( e^{X_{1}^{\varepsilon }}-e^{k}\right) ^{+}\right]  &=&E\left[
\left( e^{Z_{1}^{\varepsilon }}-e^{k}\right) ^{+}\exp \left( -\frac{kB_{1}}{%
\varepsilon \sigma }-\frac{k^{2}}{2\varepsilon ^{2}\sigma ^{2}}\right) %
\right]  \\
&=&e^{-\frac{k^{2}}{2\varepsilon ^{2}\sigma ^{2}}}e^{k}E\left[ e^{-\frac{%
kB_{1}}{\varepsilon \sigma }}\left( e^{\varepsilon \sigma B_{1}+\varepsilon
^{2}\mu }-1\right) ^{+}\right]  \\
&\sim &e^{-\frac{k^{2}}{2\varepsilon ^{2}\sigma ^{2}}}e^{k}E\left[ e^{-\frac{%
kB_{1}}{\varepsilon \sigma }}\left( \varepsilon \sigma B_{1}+\varepsilon
^{2}\mu \right) ^{+}\right]  \\
&=&e^{-\frac{k^{2}}{2\varepsilon ^{2}\sigma ^{2}}}\varepsilon \,\sigma e^{k}E%
\left[ e^{-\frac{kB_{1}}{\varepsilon \sigma }}\left( B_{1}+\varepsilon \mu
/\sigma \right) ^{+}\right] .
\end{eqnarray*}%
The justification of $\sim$ comes from a Laplace type argument: the asymptotic behaviour is determined on the event $\{ |\eps B_1| < \delta \}$, for any fixed $\delta>0$, which in turn allows to replace $e^{\eps B_1} -1$ by $\eps B_1$. (Details are left to the reader.) The next lemma, applied with $\tilde{\varepsilon}=\varepsilon \sigma /k$ (so
that $\varepsilon \mu /\sigma =$ $\tilde{\varepsilon}k\mu /\sigma ^{2}$)
then gives%
\begin{equation} \label{e:BScall}
E\left[ \left( e^{X_{1}^{\varepsilon }}-e^{k}\right) ^{+}\right] \sim e^{-%
\frac{k^{2}}{2\varepsilon ^{2}\sigma ^{2}}}\,\varepsilon ^{3}\frac{\sigma
^{3}e^{k}}{k^{2}\sqrt{2\pi }}e^{k\mu /\sigma ^{2}} \ .
\end{equation}

\begin{lemma} \label{lem:BS1}
As $\varepsilon \rightarrow 0$,%
\[
E\left[ e^{-\frac{B_{1}}{\varepsilon }}\left( B_{1}+\varepsilon \alpha
\right) ^{+}\right] \sim \frac{\varepsilon ^{2}}{\sqrt{2\pi }}e^{\alpha } \ .
\]
\end{lemma}

One can prove Lemma \ref{lem:BS1}  with an elementary Laplace type argument, noting that the relevant contribution comes from $0 < B_1 + \eps \alpha < \delta$, any $\delta > 0$. (This also explains why changing $B_1 + \eps \alpha$ by, say, $B_1 + \eps \alpha + \eps B_1 + \eps B_1^2$, will not change the asymptotics.) That said, 
Lemma \ref{lem:BS1} is also an immediate consequence of the following (non-asymptotic) estimate, which offers some flexibility in the sequel. 
\begin{lemma} \label{lem:BSabs}
Let  $\alpha \in \mathbb{R}$, $\gamma \in [0,1)$, $\varepsilon$ strictly positive, and $N \sim N(0,1)$. Then for some $C>0$, it holds that
\begin{align} \label{equ:BSabs}
  & \min\left[ (1-\gamma) e^{\frac{\alpha}{1-\gamma}} + 2 \gamma , (1+\gamma) e^{\frac{\alpha}{1+\gamma}}\right] - \varepsilon^2 \left( C(1+\alpha^2) \max\{ e^{\frac{\alpha}{1-\gamma}}, e^{\frac{\alpha}{1+\gamma}}\} + 6 \gamma\right) \notag \\
\leq  &\sqrt{2\pi} \varepsilon^{-2} \mathbb{E }\left[ \exp\left( -\varepsilon^{-1} N\right) \left( N + \gamma |N| + \varepsilon \alpha\right)^+\right]  \\
 \leq  & \max\left[ (1-\gamma) e^{\frac{\alpha}{1-\gamma}} + 2 \gamma , (1+\gamma) e^{\frac{\alpha}{1+\gamma}}\right] . \notag
\end{align}
\end{lemma}

\begin{proof}
The middle expression (\ref{equ:BSabs}) equals
$$
     \varepsilon^{-2}  \int_{-\infty}^{+\infty}  e^{-y^2/2} e^{-\frac{y}{\varepsilon}} \left(y + \gamma |y| + \varepsilon \alpha\right)^+ dy 
=   \int_{-\infty}^{+\infty}   e^{-v^2 \eps^2 /2}  e^{-v} \left(v + \gamma | v| +  \alpha\right)^+ dv.
$$
The elementary $1 - y^2/2 \le \exp(-y^2/2) \le 1$ then leads to the stated bounds. Indeed,
$$
           \text{(\ref{equ:BSabs})} \leq   \int_{-\infty}^{+\infty}   e^{-v} \left(v + \gamma | v| +  \alpha\right)^+ dv
$$
which is computable: when $\alpha < 0$, the right-hand side equals
$$
(1+\gamma) \int_{-\frac{\alpha}{1+\gamma}}^{\infty} \left( v + \frac{\alpha}{1+\gamma}\right)  e^{-v} dv = (1+\gamma) e^{\frac{\alpha}{1 +\gamma}}
$$
whereas when $\alpha \geq 0$ we have
\begin{align*}
&\int_{-\frac{\alpha}{1-\gamma}}^{0}   e^{-v} \left(v (1-\gamma) +  \alpha\right)dv + \int_0^\infty e^{-v} \left(v (1+\gamma) +  \alpha\right)  dv\\
=& (1-\gamma) e^{\frac{\alpha}{1-\gamma}} + \int_0^{\infty} e^{-v}2 \gamma v dv = (1-\gamma) e^{\frac{\alpha}{1-\gamma}} + 2 \gamma \;.
\end{align*}
To obtain the lower bound, use $e^{-y^2/2} \geq 1 - y^2/2$ and split the integral to obtain
\[
 \text{(\ref{equ:BSabs})} 
\geq   \int_{-\infty}^{+\infty}   e^{-v} \left(v + |\gamma| v +  \alpha\right)^+ dv  - \frac{\varepsilon^{2}}{2} \int_{-\infty}^\infty e^{-v} v^2 \left(v + |\gamma| v +  \alpha\right)^+ dv.
\] 
The first integral is computed as before. For the second one, we again distinguish according to the sign of $\alpha$. When $\alpha <0$ we obtain 
$$ e^{\frac{\alpha}{1+\gamma}}\int_{0}^{\infty} e^{-u} \left(u - \frac{\alpha}{1+\gamma}\right)^2 u du \leq C  e^{\frac{\alpha}{1+\gamma}}(1+\alpha^2) $$
whereas in the case $\alpha \geq 0$ we have
\begin{align*}
&\int_{-\frac{\alpha}{1-\gamma}}^{0}   e^{-v}v^2  \left(v (1-\gamma) +  \alpha\right)dv + \int_0^\infty e^{-v} v^2 \left(v (1+\gamma) +  \alpha\right)  dv\\
=& e^{\frac{\alpha}{1-\gamma}} (1-\gamma) \int_0^{\infty}  e^{-u} \left(u - \frac{\alpha}{1+\gamma}\right)^2 u du + \int_0^{\infty} e^{-v} 2 \gamma v^3 dv \\
\leq & C  e^{\frac{\alpha}{1-\gamma}}+ 12 \gamma.
\end{align*}
\end{proof}

\section{Unified large, moderate and rough deviations} \label{sec:uniLDP}

We now put forward our basic large deviation assumption. The object of interest is a scalar process $X^\eps \equiv X_{\eps^2 (\cdot)}$, viewed as (normalized) log-price process run at speed $\eps^2$. By convention, $X_0 = 0$ so that $X^\eps_1 \to 0$ as $\eps \to 0$.
The reader can have in mind an It\^o diffusion: classical StochVol models assume that $(X^\eps)$ is one component of a higher-dimensional Markov diffusion; RoughVol models have additional components driven (or given) by fractional Brownian motion. We further note that our setup includes pricing in the moderate deviation regime.  

\subsection{Basic large deviation assumption (A1)} 

\begin{itemize}
\item[(A1a)] The family $\{ \bar X^\eps_1 : 0 < \eps \le 1 \}$, given as 
         unit time marginal of a rescaled process $\X^\eps$,
         $$\bar X^\varepsilon := \frac{ \beps }{\eps} X^\varepsilon$$ satisfies a LDP with good rate function (a.k.a. {\it energy}) $\I$ and speed $\beps^2 \ge \eps^2${for some given function $\eps \mapsto \beps = \beps(\eps)$.}
\item[(A1b)] Assume $\I(x)=0$ iff $x=0$.
\item[(A1c)] The (in general only lower semicontinuous) rate function $\I$ is continuous.
\end{itemize}
Note that a large deviation principle for $(X^\varepsilon)$ may or may not hold. Remark that condition (A1b) says $\I(0) = 0$ (which encodes $X^\eps_1 \to 0$) and rules out moves to $ x \ne 0$ at zero cost.
 (A1c) is known in all the StochVol / RoughVol setting interesting to us, cf. references below.\footnote{The forthcoming non-degeneracy condition (A5) implies local $C^1$-regularity of $\I$, near a given point $x$.}

\begin{example}[Black-Scholes / Schilder LDP]
Let $X_1^\eps = \eps \sigma B_1 + \eps^2 \mu$ (with $\mu = - \sigma^2 /2$ in absence of rates) and a LDP holds with 
$$
                    \beps =  \eps , \ \ \  \I(x) = \tfrac{x^2}{ 2\sigma^2} \ .
$$
\end{example}

\begin{example}[Classical StochVol, Freidlin-Wentzell LDP] Here $\I$ is in general non-explicit, but has an interpretation in terms of geodesic distance from arrival log-spot/spot-vol ($0,y_0)$ to the arrival manifold $(x,\cdot)$.  In a locally elliptic setting, the rate function is viscosity solution to eikonal equation hence continuous, see e.g. \cite[Thm 2.3]{berestycki2004computing}. It is shown in \cite{deuschel2014marginala} that $\I$ is smooth ``away from focality points'' which is always the case for $x$ close to zero.

\end{example}

\begin{example}[RoughVol, Forde-Zhang LDP \cite{forde2017asymptotics}] \label{ExRVLDP}
Let $H \in (0,1/2]$. With 

\begin{equation*}
X_1^{\varepsilon }=\int_{0}^1 \s \left( \hat \eps \hat W \right) \varepsilon
d\left( \bar{\rho} \bar W+\rho W\right) -\frac{1}{2}\varepsilon ^{2}\int_{0}^1
\s^2 \left(  \hat \eps \hat W\right) dt \ ,
\end{equation*}%
where $\hat W = K^H * \dot W$, $\hat \eps := \eps^{2H}$, the LDP assumption holds for  
$$\bar X^\eps_1 = \int_0^1 \s \left( \hat \eps \hat W \right) \hat\varepsilon
d\left( \bar{\rho} \bar W+\rho W\right) -\frac{1}{2}\varepsilon \hat{\eps}\int
\s^2 \left(  \hat \eps \hat W\right) dt \ ,
$$
i.e. with $\beps = \hat \eps =  \eps^{2H}$, and with continuous \cite[Cor. 4.6.]{forde2017asymptotics} rate function
\begin{equation}\label{ifz}
          \I(x) = \mathcal{J}(x) := \inf \{ \tfrac{1}{2}\| h,\bh \|^2_\H :  \int_{0}^1 \s \big(\, \hat h \, \big) 
d\left( \bar{\rho} \bh+\rho h \right) = x \} \equiv \tfrac{1}{2}\| h^x,\bh^x \|^2_\H \ .
\end{equation}      
\end{example}
We note that \cite{forde2017asymptotics} assumes a (linear) growth condition on $\sigma(.)$. It was seen in \cite{bayer2019regularity}, see also \cite{jacquier2018pathwise, gulisashvili2017large} and assumption (A3) below, that condition (A1a) holds without  growth condition of $\sigma(.)$ and hence is valid for the exponential function and then the RoughBergomi model, cf. Section \ref{sec:rmkRB}.
Condition (A1c) follows from \cite[Cor. 4.6.]{forde2017asymptotics}.

 

\begin{example}[Black-Scholes MDP]
Take $\beta \in [0,1/2)$ and $\beps =  \eps^{1-2\beta}$. Then the LDP assumption holds, with speed $\beps^2$ and rate function $\I(x) = \tfrac{x^2}{ 2\sigma^2}$, for 
$$
\bar X_1^\eps = \beps \sigma B_1 + \beps \eps \mu = \frac{ \beps }{\eps} ( \eps \sigma B_1 + \eps^2 \mu) \ .
$$
\end{example}

\begin{example}[Classical StochVol, moderate deviations regime \cite{friz2017option}] With  $\beta \in (0,1/2)$ 
and $\beps =  \eps^{1-2\beta}$ as above, one finds that $\bar X_1^\eps = \frac{ \beps }{\eps} X_1^\eps$ 
satisfies the same MDP as Brownian motion, with Black-Scholes paramater $\sigma$ replaced by 
spot-volatility $\sigma_0$. For $\beta=0$ we are in the realm of Freidlin-Wentzell LDP, $\beta = 1/2$ is the central limit scaling.
\end{example}

\begin{example}[RoughVol, moderate deviations regime \cite{bayer2017short}] \label{ex:modrough}
Consider the log-price $X^\eps$ under RoughVol, as introduced in Example \ref{ExRVLDP}. 
Let $0 < \beta < H \le 1/2$ and set $\beps =  \eps^{2H-2\beta}$. Then the LDP assumption holds, with speed $\beps^2$ and rate function $\I(x) = \tfrac{x^2}{ 2\sigma_0^2}$, for 
$$
\bar X_1^\eps = \frac{ \beps }{\eps} X_1^\varepsilon \ .
$$
For $\beta = 0$, we are back in the large deviation setting, $\beta= H$ is the central limit scaling. 

\end{example}

\subsection{Moment assumption (A2) and large deviation option pricing} \label{sec:LDOP}

Consider $$0 \le k_\eps = (\eps / \beps) \, x \le x \ . $$ 
{ Under assumption (A1a) we have} 
\begin{equation} \label{eq:ldpXbar}
           P [ \bar X_1^\eps > x ] = P [ X_1^\eps > k_\eps ] = \exp \left( -  \frac{\I(x) + o(1) }{\beps^2}    \right) \ ,
\end{equation}
{ which we view as (out-of-the money, short: OTM ) digital option price, with underlying log-price $X_1^\eps$ and log-strike $k_\eps$. As we will see below, the situation is similar for put prices. However, in the case of call-prices of the form} 
$$
                 E [ (\exp( X_1^\eps) - \exp (k_\eps))^+]  
$$
{ the unbounded payoff requires the following moment assumption:}
\smallskip
\begin{itemize}
          \item[(A2)]
         There exists $p>1$ such that $\limsup_{\eps \to 0} E [ (\exp( X_1^\eps)^p ] =: m_p < \infty$
\end{itemize}
\smallskip
\noindent Here is a typical way to check (A2). The proof is left to the reader.
\begin{lemma} { Assume that there exists a process $(X_t)_{t \geq 0}$ with $X_t \overset{(d)}{=} X_1^{t^{1/2}}$} such that $e^{-rt} \exp (X_t) \equiv S_t$ is a martingale, and and assume there exists $p>1$ and $t>0$ such that 
$E [ S_t^p ] <\infty$. Then (A2) holds.
\end{lemma} 

{
 \begin{proposition}[Small noise OTM pricing] \label{prop:callLDP}
Set $k_\eps = (\eps / \beps) \, x$. Assume (A1). Then we have
$$
       \text{Put prices asymptotics, $x<0$: } E [ (\exp (k_\eps)-\exp( X_1^\eps) )^+]  = 
       \exp \left(  - \frac{\I(x)+o(1)}{\beps^2}    \right) \ .
$$
Under the additional moment assumption (A2) we also have the 
$$
        \text{Call prices asymptotics, $x>0$: } E [ (\exp( X_1^\eps) - \exp (k_\eps))^+]  = 
       \exp \left(  - \frac{\I(x)+o(1)}{\beps^2}    \right) \ .
$$
\end{proposition}

}

\begin{remark} \label{rem:A2}
(Condition A2) (i) Existence of a first moment ($p=1$) is a ``conditio sine qua non'' for martingale based pricing theory. Existence of ``$1^+$ moments'' is hence a fairly mild and frequently encountered condition, satisfied by the bulk of models used in practice. This includes RoughVol \`a la Forde--Zhang under a growth restriction on $\sigma(.)$ as well as the RoughHeston model. That said, non-existence of $1^+$ moments, and even loss of martingality, are known to occur for models with log-normal volatility \cite{sin1998complications, jourdain2004loss, lions2007correlations}. As shown in these references, in classical StochVol settings, condition (A2) holds if (and only if) the driving factors are $\rho$-correlated with $\rho < 0$. Fortunately, this assumption of negative correlation is well-justified from a financial perspective, which explains why models such as the ``log-normal'' SABR model are in widespread use. At this moment, validity of (A2) for (negatively correlated) rough Bergomi model (cf. Example \ref{ExRVLDP}, think: $\sigma(.)$ exponential) remains a  conjecture. The difficulties are described, together with decisive results on martingality and the ``only if''-part, in \cite{gassiat2018martingale, 2018arXiv180800421G}. Note however that the ``Put price'' part of Proposition \ref{prop:callLDP} does not rely on (A2) and hence applies to rough Bergomi.\footnote{We thank a referee for prompting us to make this application to Rough Bergomi explicit.} \\

(ii) Proposition \ref{prop:callLDP} improves on \cite[Corollary 4.9]{forde2017asymptotics} where similar call price asymptotics are obtained under a linear growth condition of $\sigma$ (or - tracking their proof - under the assumption of arbitrarily high finite moments, for $\eps$ small enough; cf. also  \cite{bayer2017short}). For instance, this includes Heston but rules out ``log-normal'' SABR. See also \cite{2018arXiv180800421G} for a recent discussion on such moment conditions.
\end{remark}

\begin{proof} {The details of all cases other than the upper call price bound are straight-forward and spelled out in \cite[Appendix C]{forde2017asymptotics}, hence omitted.}
Fix $y > x$ and consider $\eps \in (0,1]$ so that $\beps / \eps = \nu_\eps \ge 1$ and then
\begin{eqnarray*}
E [ (\exp( X_1^\eps) - \exp (k_\eps))^+]  &=&  E [ (\exp(  \bar X_1^\eps / \nu_\eps) - \exp (x/\nu_\eps))^+ 1_{\{ \bar X_1^\eps \in (x,y]\}}] \\
& & +
  E [ (\exp( X_1^\eps ) - \exp (k_\eps))^+ 1_{\{ \bar X_1^\eps > y \}}]  \\
  &\leq& (e^{y/\nu_\eps}-e^{x/\nu_\eps})^+ P [ \bar X_1^\eps > x ] + E [ (\exp( X_1^\eps)^p ]^{1/p} P( \bar X_1^\eps > y)^{1/q}
   \\
  &\leq& (e^{y}-e^{x})^+ P [ \bar X_1^\eps > x ] + E [ (\exp( X_1^\eps)^p ]^{1/p} P( \bar X_1^\eps > y)^{1/q}
\end{eqnarray*}

where we have taken H\"older's inequality for the second term, thanks to (A2), with H\"older conjugate $q = p'  < \infty$, and, uniformly over small $ \eps$, with
 $E [ (\exp( X_1^\eps)^p ]^{1/p} < \infty$.
Use (A1) to obtain
$$  \limsup_{\eps \to 0} \bar{\eps}^2 \log \left(E [ (\exp( X_1^\eps) - \exp (k_\eps))^+]\right) \leq \max \left( - \I(x) , - \frac{\I(y)}{q}\right)$$
and we conclude by letting $y \to +\infty$.
(Thanks to goodness of the rate function, one cannot reach infinity at finite cost so that $\I(y) \to \infty$ as $y \to \infty$.) 
\end{proof}
 
 When applied to classic StochVol option pricing in the short-time limit, $t= \eps^2$, Proposition \ref{prop:callLDP} is a rigorous formulation of what is often loosely written as
 $$
                 E\big[  (S_t - K)^+ \big] \approx e^{-\I(k)/ t}
 $$
 where $\I(k)$ is the rate function, $k = \log (K / S_0)$. Similarly, under RoughVol, with energy function $\I$, this relation becomes
 $$
                 E\big[  (S_t - K_t)^+ \big] \approx e^{-\I(x)/ t^{2H}}
 $$
 where $K_t = S_0 \exp ( x t^{1/2-H} )$. In the corresponding moderate regime one has
 $$
                 E\big[  (S_t - L_t)^+ \big] \approx \exp \left( - \frac{x^2}{2\sigma_0^2 t^{2H-2\beta}} \right) ,
 $$
 with $L_t = S_0 \exp ( x t^{1/2-(H-\beta)} )$ where $0< \beta < H < 1/2$. (The ``moderate'' approximation formula under classical StochVol is exactly of the same form, with $H=1/2$.)
 
 \medskip 
 
 The remainder of this paper is devoted to replace $\approx$ by a honest asymptotic equivalence, as seen in the 
Black-Scholes example: for fixed $x>0$, as $t \downarrow 0$,
$$ 
 E [ (e^{\sigma B_t + \mu t} - e^x)^+ ]  \sim e^{-x^2 / (2 \sigma^2 t )} t^{3/2} \ \tfrac{ \sigma^3  e^x e^{x \mu / \sigma^2} }{ x^2 \sqrt{2\pi} } \ .
$$
As we shall see, our methods apply in great generality to obtain such ``precise'' large deviation for StochVol and RoughVol. A detour is necessary in the moderate regime where the presence of  another scale $\eps^\beta$ rules out the stochastic Taylor expansions in $\beps = \eps^{2H}$. Nonetheless, we have precise enough control, that the moderate expansion is obtained, in essence, from uniformity of our precise large deviation estimates. 

\section{Examples, robustness and non-degeneracy} 
We are interested in refining the large deviation option price asymptotics of Proposition \ref{prop:callLDP}. To this end, we have to be more explicit about the construction of the $X^\eps$. The basic setup, is an $m$-dimensional Wiener space $C([0,1],R^m)$, 
equipped with the Wiener measure, as common underlying probability space for the processes $(X^\eps)$, constructed, for instance, as strong solution to stochastic It\^o or Volterra differential 
equations. In all these examples, $(X^\eps)$ has somewhat obvious meaning when the white noise $\dot \W$ is replaced (or perturbed) by some Cameron--Martin element $\dot \h \in L^2([0,1],R^m)$ -- which in turn underlies large deviation (and precise) asymptotics. If one wants to be less specific about the origin of the Wiener functionals under consideration, exhibiting the right abstract condition is no easy matter and leads to various notions of ``regular'' Wiener functionals \cite{kusuoka1991precise} (applied to option pricing in \cite{kusuoka2008remark,osajima2015general}). Our abstract assumptions, cf. (A3) and (A4) are different: in a sense we avoid working with a regular Wiener functionals by imposing ``robustness'' in the sense of stability w.r.t. a suitable enhancement of the noise which restores analytical control.{In order to motivate and explain assumptions, in the next two sections we consider two fundamental classes of stochastic volatility models: classical StochVol and rough StochVol. Following such examples, in Section \ref{sec:ass345} we state our robustness and control assumptions.}

\subsection{The case of SDEs and (classical) StochVol}  \label{sec:SDEs}

Consider an $n$-dimensional diffusion given in strong It\^o SDE form, in the small noise regime
$$
               d \X^\eps = b ( \eps, \X^\eps) dt + \sigma ( \X^\eps) \eps d \W \ ,
$$
with fixed initial data, driven by an $m$-dimensional Brownian motion $\W$. The asymptotic analysis of such small noise equations is classical, see e.g. \cite{azencott1982formule, azencott1985petites, arous1988methods}.

In a (classical) StochVol setting, $\X^\eps = (X^\eps, Y^\eps) \in \R \times \R^{n-1}$ where $X^\eps$ has the interpretation of price (or log-price) process, together with $(n-1)$ volatility factors.  Fix $\eps>0$. Under standard assumptions on $b, \sigma$ there exists a measurable map (``It\^o-map'') $ \phi^\eps$ on Wiener space, such that
$$
               X^\eps (\omega) = \phi^\eps (  \eps \W(\omega)) \ .
$$
Provided that $ b ( \eps, \cdot) $ converges (uniformly on compacts) to $b (0, \cdot)$ as $\eps \downarrow 0$, the family $(\X^\eps )$ is exponentially equivalent to $ \phi^0 (  \eps W(\omega))$; a small noise LDP in this setting is provided by Freidlin-Wentzell theory, with (good) rate function given by formal (!) application of the contraction principle and Schilder's LDP for Brownian motion. Stochastic Taylor expansions of the form
\begin{equation}
          \phi^\eps_1 (  \eps \W(\omega) + \h) =  \phi^0_1 (\h) + \eps g_1 ( \omega) + \eps^2 g_2 ( \omega) + \dots + \eps^{n-1} g_{n-1} ( \omega)  + r_n^\eps (\omega) \label{equ:STE}
\end{equation}
are easily obtained by a formal computation (and can subsequently be justified, including remainder estimates). In Lyons' rough path theory 
one enhances noise $\W$ to 
a (random) {\it rough path} of the form $\WW = (\W, \mathbb{W}) = (\W, \int \W \otimes d\W) \in \C^\alpha, \alpha \in (1/3,1/2)$, with It\^o integration, such that
$$
         \phi^\eps (  \W(\omega)) =  \Phi^\eps (  \WW (\omega)) 
$$
where $ \Phi^\eps$ is the It\^o-Lyons map, known to be locally Lipschitz continuous with respect to suitable rough path metrics. We also recall that
$$ \mathbb{W} \equiv \mathbb{W}^{\mathrm{Ito}} \ne \mathbb{W}^{\mathrm{Strato}} = \lim_{\eta \to 0} \int \W^\eta \otimes d\W^\eta \;,$$ where the limiting statement is a well-known Wong-Zakai\footnote{By $\W^\eta$ we mean a mollification at scale $\eta \to 0$, convergence is in probability and $ \C^\alpha$.} statement for the (Stratonovich) Brownian rough paths, see e.g. \cite[Ch. 3]{friz2014course}; with identity matrix $I$, the It\^o-Stratonovich correction reads
\begin{equation} \label{equ:ISSDE}
           \mathbb{W}_{s,t} = -\frac{I}{2} (t-s) + \lim_{\eta \to 0} \int \W^\eta \otimes d\W^\eta \;.
\end{equation}
Scaling $\eps \W$ lifts to {\it rough path dilation} $\delta_\eps \WW := (\eps \W, \eps^2 \int \W \otimes d\W)$ so that $\phi^\eps (  \eps \W(\omega)) =  \Phi^\eps ( \delta_\eps \WW (\omega))$. 
(A LDP for $\delta_\eps \WW$ in rough path topology then provides an easy proof of Freidlin-Wentzell large deviations for SDEs, see e.g. \cite[Ch.9]{friz2014course}.) 
Last but not least, Cameron-Martin translation $\W \mapsto \W + \h$ lifts to {\it rough path translation}, formally given by $T_{\h} \WW = (\W + \h, \int (\W+\h) \otimes d(\W+ \h)) $, so that, for fixed $\h\in\CM$ and any $n=1,2,...$,
\begin{equation} \label{equ:STSDE}
       \phi^\eps_1 (  \eps \W(\omega) + \h)  = \Phi^\eps_1 (T_{\h} (\delta_\eps \WW) )  = \Phi^0_1 (\h) + \eps G_1 ( \WW) + \dots + \eps^{n-1} G_{n-1} ( \WW) + R_n^\eps (\WW) \ .
\end{equation} 
We insist that the second equality is a purely deterministic ``rough Taylor expansion'', which becomes random only after inserting the random rough path $\WW = \WW (\omega)$, constructed from Brownian motion. The point of this construction (developed in \cite{aida2007semi, inahama2007asymptotic, inahama2008laplace, inahama2010stochastic}) is the robustification of all terms in (\ref{equ:STE}), i.e.
\begin{equation} 
      g_1 ( \omega) = G_1 ( \WW (\omega)), \ g_2 ( \omega) = G_2 ( \WW (\omega)), \dots , r_n^\eps (\omega) = R_n^\eps (\WW (\omega)) \ .
\end{equation} 
with precise (deterministic) rough path estimates, notably  (cf. \cite[Thm. 5.1]{inahama2007asymptotic})
\begin{equation}
          |   R_n^\eps (\WW) | \lesssim \eps^n ( 1 + ||| \WW |||^n ) \ ,   \label{equ:Inahama}
\end{equation}
valid on bounded sets of $\delta_\eps \WW$, uniformly over $\eps \in (0,1]$. As will be seen, expansion to order $2$ is sufficient for our purposes. Here $||| \WW ||| \equiv \| W \|_\alpha + \| \mathbb{W} \|^{1/2}_{2\alpha}$ is the homogenous rough path norm, as found e.g. in \cite[Ch. 2]{friz2014course}; the prototypical example of a homogeneous model norm, cf. Appendix. 

We need to understand the behaviour of $G_1,G_2$ under scaling and translation. The situation is somewhat simplified by assuming that the drift vanishes to first order, as is typical in the setting of 
short-time asymptotics (think: $b(\eps,x) = \eps^2 b (x)$), i.e. 
%
 $\partial_\eps b(0,\cdot) = 0 $. Then $\partial_\eps \Phi (0,\cdot) \equiv \partial_\eps |_{\eps=0 }\Phi^\eps (\cdot) = 0$ and, for all
$\eps \in (0,1]$, $\h \in \H$, and $ \WW \in \mathcal{C^\alpha}$,
\begin{equation} \label{equ:STSDE0}
  \Phi^0_1 (T_{\h} (\delta_\eps \WW) )  = \Phi^0_1 (\h) + \eps G^0_1 ( \WW) + \eps^2 G^0_2 (\WW) + R_3^{0,\eps} (\WW)
\end{equation} 
with continuous linear $G^0_1 := G_1$ and continuous quadratic  $$G^0_2 := G_2 - \frac{1}{2} \partial^2_{\eps \eps} \Phi (0,\h)$$ and remainder estimate for  $ R_3^{0,\eps} (\WW)$ exactly as in (\ref{equ:Inahama}).
Here we call a map $G:  \mathcal{C}^\alpha \mapsto C[0,1]$ {\it continuous linear} if there exists a continuous linear map $G_1: C^\alpha \to C[0,1]$ so that $G(\WW) = G_1(\W)$. (As a trivial consquence,  $\sup \{ G (\WW) : ||| \WW ||| \le 1  \}  < \infty$ and
$G (T_\k \delta_\varepsilon \WW) = G( \kk) + \varepsilon G (\WW)$.)
Similarly, call a continuous map $G: \WW \mapsto C[0,1]$ {\it continuous quadratic} if $\sup \{ G (\WW) : ||| \WW ||| \le 1  \}  < \infty$ and 
$$
                 G (T_\k \delta_\varepsilon \WW) = G ( \kk) + \varepsilon  G (\WW,\kk) + \varepsilon^2 G (\WW) 
$$
where $G (\WW,\kk)= G_2(\W,\k)$, for a continuous bilinear map $G_2 :  C^{\alpha} \times \CM \to C[0,1]$.

\begin{example}[Black-Scholes]
Let $X_1^\eps =\phi_1 (\eps, \eps B)$ with $\phi_1 (\eps, B) = \eps^2 \mu +  \sigma B_1$. (There is no need for a rough path lift here.) After a Girsanov shift $\eps B \to \eps B + h$, have 
$$ Z_1^\eps :=\phi_1 (\eps, \eps B + h ) = \sigma h_1 + \eps \sigma B_1 + \eps^2 \mu \ ,$$
which allows to read off $g_1,g_2$ (and hence $G_1,G_2$, in obvious ``pathwise'' robust form) and zero remainders. Note $\phi_1 (\eps, h ) = \sigma h_1 + \eps^2 \mu$, so that $\frac{1}{2} \partial^2_{\eps \eps} \phi_1 (0,h) = \mu$. (There is no need here to distinguish between $\phi$ and $\Phi$). In particular, $G_2^0 \equiv 0$.
\end{example}

\begin{example}[Stein-Stein]
For given parameters, $a\geq 0,b \le 0,c>0,\sigma _{0}\geq 0$, the Stein--Stein model
expresses log-price $X$, with $\rho$-correlated Brownian drivers $W, \bar W$,
\begin{eqnarray}
dX^\eps &=&-\frac{1}{2}(Y_t^\eps)^{2}\eps^2dt+Y_t^\eps \eps d\tilde W,\,\,X^\eps\left( 0\right) =x_{0}=0
\label{SteinSteinSDE} \\
dY^\eps &=&\left( a+bY^\eps\right) \eps^2dt+c \eps dW,\,\,\,Y^\eps\left( 0\right) =\sigma _{0}>0. 
\notag
\end{eqnarray}
In terms of a two-dimensional standard Browninan motion $(W, \bar W)$, we may take $\tilde W = \bar{\rho} \bar W+\rho W$.
For simpler expressions, take $a=b=0$ so that
$$
        X_1^\eps = -\frac{1}{2}\int_0^1(\sigma_0 + c\eps W_t)^{2}\eps^2dt + \int_0^1(\sigma_0 + c\eps W_t) \eps d \tilde W_t  \; .
$$
This is a rare and instructive example where the robust form can be seen explicitly: $X_1^\eps$ is an explicit function of the ($\eps$-dilated) Brownian rough path at unit time. Indeed $\delta_\eps \WW_1$ manifestly contains $\eps W_1, \eps \bar W_1,\eps^2 \int_0^1 W dW, \eps^2 \int_0^1 W d\bar W$, and $X_1^\eps$ is a second order polynomial in $\eps$ and these four expressions. As a consequence, also after some translation $W \to W + h, \bar W \to \bar W + \bar h$, we have an expansion (\ref{equ:STSDE}) with non-trivial $G_1$, $G_2$ but zero remainder $R_3^{\eps} \equiv 0$.  
It is not difficult to modify this example such as to see $\eps^3$ (and higher) in the stochastic Taylor expansion. Then there will be a non-vanishing remainder and, in general, we cannot write things down explicitly.
\end{example}

\bigskip

\subsection{The case of rough volatility}  \label{sec:roughDEs}

Following the notation of Example \ref{ExRVLDP}, we consider the small noise setting for rough volatility. Let $\sigma(.)$ be a scalar function. There is interest in avoiding too restrictive growth conditions on $\sigma(.)$ such as to include exponential functions (cf. RoughBergomi \cite{bayer2016pricing}), whereas one can safely assume $\sigma(.)$ to be smooth. Recall $\hat \eps = \eps^{2H}, H \in (0,1/2]$. Recall that
$$
          \W = (W, \bar W)
$$
consists of two independent (standard) Brownian motions. These are used to construct 
$$
      \tilde W = \bar{\rho} \bar W+\rho W \ \ \ \text{and} \ \ \    \hat W = K^H * \dot W  \ ,
$$
so that $\tilde W$ is again a standard Brownian motion ($\rho$-correlated with $W$) whereas $\hat W$ is a fractional Brownian motion, only dependent on $W$. 
Note that $\tilde W = \bar W$ in the uncorrelated case; note also $\hat W = W$ in the case of $H=1/2$, which falls under the classical StochVol setting. We will use identical notation when dealing with Cameron--Martin paths
$ \h = (h, \bar h)$, so that $\tilde h = \bar{\rho} \bar h+\rho h, \ \    \hat h = K^H * \dot h$. Under rough volatility, the (rescaled)  log-price process has the form
$$ \bar X^\eps_1 = \int_0^1 \s \left( \hat \eps \hat W \right) \hat \eps
d \tilde W  -\frac{1}{2}\eps \hat{\eps}\int_0^1
\s^2 \left(  \hat \eps \hat W\right) dt \ .
$$
A stochastic Taylor expansion, after a Girsanov shift $\hat \eps (W, \bar W) \mapsto  (\hat \eps W + h , \hat \eps \bar W + \bar h) $, gives, 
\begin{eqnarray}
\bar Z^{\eps, \h}_1 &:=&\int_0^1\s \left( \hat\eps \hat W+\hat h\right)
d[ \hat\eps \tilde W + \tilde h ]
-\frac{\eps \, \hat\eps}{2} \int_0^1\sigma ^{2}\left( \hat\eps
\hat W+\hat h\right) dt \nonumber \\
& \equiv & g_0 + \hat \eps g_1 (\omega) + \hat \eps^2 g_2 (\omega) + r_3(\omega) \ .  \label{equ:defr3}
\end{eqnarray}%
(We should really write $g_i^\h$ and similarly for $r_3^\h$, but later $\h$ will be fixed.)
We can read off $g_0 $ and $g_1 (\omega)$ as zero and first order terms (in $\hat\eps$) of the expansion
\begin{eqnarray}
\int_0^1\sigma \left( \hat\eps \hat W+\hat h\right)
d[ \hat\eps \tilde W + \tilde h ] & = &  \int_0^1\sigma \left( \hat h\right) d \tilde h
+\hat\eps \left( \int_0^1\sigma \big(\, \hat h \, \big) d\tilde W+\int_0^1\sigma ^{\prime }\big(\, \hat h \, \big) \hat W
d \tilde h
\right)  \nonumber \\
& & +\hat\eps^2 \left( \int_0^1\sigma' \big(\, \hat h \, \big) \hat W d\tilde W+ \frac{1}{2} \int_0^1\sigma'' \big(\, \hat h \, \big) \hat W^2 d \tilde h \right) 
+ \dots \label{equ:dotsCaseHonehalf}
\end{eqnarray}%
but the precise form of $g_2(\omega)$  - and thus the remainder $r_3(\omega)$ implicitly defined in (\ref{equ:defr3}) -  requires the following distinction:

\medskip

\noindent {\bf Case of $H=1/2$.} In this case $ \hat \eps =  \eps, \hat h =h $ and so
\begin{equation}\label{def:g2}
         g_2(\omega) = \int_0^1\sigma' (h) \hat W d\tilde W+ \frac{1}{2} \int_0^1\sigma'' ( h) \hat W^2 d \tilde h - \frac{1}{2} \int_0^1\sigma ^{2} (  h) dt \\
\end{equation}
Remark that in this case the remainder $r_3(\omega)$ is the sum of the implicit (dots) remainder in (\ref{equ:dotsCaseHonehalf}) and $\eps^2$ times the difference
$$
        \frac{1}{2} \int_0^1\sigma ^{2}\left( \eps W+ h\right) dt -  \frac{1}{2} \int_0^1\sigma ^{2} (  h) dt = O (\eps \| W \|_{\infty;[0,1]} ) \ , 
$$
valid on
$\{ \omega : \eps \| W (\omega) \|_{\infty;[0,1]} < \delta \}$, for any finite $\delta$, which shows that this term contributes $C\eps^3 \| W \|_{\infty;[0,1]}$ to the remainder, where $C$ can be taken as $C^1$-norm of $\sigma^2$, restricted to a $\delta$-fattening of the set $\{ h(t): 0 \le t \le1 \}$. 

\medskip

\noindent
{\bf Case of $H < 1/2$.} In this case the second order (in $\hat \eps)$ term is given by
\begin{equation}\label{def:g2rough}
          g_2(\omega) = \int\sigma' \big(\, \hat h \, \big) \hat W d\tilde W+ \frac{1}{2} \int_0^1\sigma'' \big(\, \hat h \, \big) \hat W^2 d \tilde h \\
\end{equation}
whereas the remainder $r_3(\omega)$ is the sum of the implicit (dots) remainder in (\ref{equ:dotsCaseHonehalf}) and
$$
     -\frac{\eps \, \hat\eps}{2} \int_0^1\sigma ^{2}\left( \hat\eps
\hat W+\hat h\right) dt = O(\eps \, \hat\eps) = o (\hat \eps^2) \ ,
$$
valid on $\{ \omega : \eps \| W (\omega) \|_{\infty;[0,1]} < \delta \}$, as before. 
%
%

\medskip

{\bf Robust form.} The noise $\W = \W (\omega)$ can be lifted to a (random) model $$ \WW = \left(\W, \int \hat W dW, \int \hat W^2 dW , ....\right) \in \M$$ 
with finitely many ``building blocks'', the first few of which have been listed. The (log)price process under rough volatility can then be written as continuous image $\Phi(\WW)$, see 
 \cite{bayer2019regularity}. It is an important fact that Wong-Zakai approximations to this model fail to converge (due to $H<1/2$), but at the price of
 subtracting the correct (diverging) quantities one does obtain the above model (defined in term of It\^o integration) as limit; cf. Appendix \ref{app:ERS}.
 (Although the ``It\^o model" of interest is given by standard left-point approximations, such renormalized Wong-Zakai approximations will come in handy later on, cf.Lemma \ref{lem:AboutV}.) The Cameron-Martin space acts naturally by translation. For any $\h \in \H$, we have natural map $T_{\h}: \M \to \M$, with inverse $T_{-\h}$, which ``lifts" the meaning of $\W \mapsto \W + \h$, with an estimate (Lemma \ref{lem:RVMtranslation}), 
of the form $$||| T_{\h} \WW ||| \lesssim ||| \WW ||| + \| h ||_{\CM}. $$
 This translation is moreover  consistent with Cameron--Martin shift of Wiener paths in the sense that
$$
            \forall \h \in \CM: \WW (\omega + \h) = T_{\h} \WW (\omega) \text{ a.s.} \ .
$$
 One further defines dilation of the above model, similar to rough paths, by 
\begin{equation}\label{def:dilation}
          \delta_\eps \WW = \left(\eps \W, \eps^2 \int \hat W dW,  \eps^3 \int \hat W^2 dW , ....\right) \in \M \ .
\end{equation}
We then have 
$$
          \bar Z^{\eps, \h} (\omega) =  \Phi^\eps ( T_\h \delta_{\hat \eps} \WW (\omega) ) \ ;
$$
where we introduced, for an arbitrary $\MM \in \M$,
$$
          \Phi^0 ( \MM ) :=  \int \s \left( \hat M \right)  d \MM \ , 
$$
and then, for every  $\eps \ge 0$, 
$$
                 \Phi^\eps ( \MM )  \equiv \Phi( \hat \eps, \MM ) := \Phi^0 ( \MM ) - \frac{\eps \, \hat\eps}{2} \int \sigma ^{2}\left( \hat M \right) dt \ .
$$
Observe that $\partial_{\hat \eps} \Phi (0,\cdot) \equiv 0$. We see in Theorem \ref{thm:STLE} that
$$
                 \Phi^0_1 ( T_\h \delta_{\hat \eps} \WW (\omega) ) =   \Phi^0_1 (\h) +  \hat \eps G^0_1 (\WW (\omega)) +   \hat \eps^2 G^0_2 (\WW (\omega)) + R^0_3 (\WW (\omega)) 
$$
where $G^0_1, G^0_2$ are continuous linear and quadratic, respectively, in the obvious adaption of these notions 
to the present setting of regularity structures and
$$
    | R^0_3 (\WW ) | \lesssim \hat \eps^3 ||| \WW |||^3 \ 
$$ 
on bounded sets of $\hat \eps ||| \WW |||$. Mind that the corresponding expansion 
$$ \Phi^\eps_1 ( T_\h \delta_{\hat \eps} \WW (\omega) ) = \Phi^0_1 (\h) +  \hat \eps G_1 (\WW (\omega)) +   \hat \eps^2 G_2 (\WW (\omega)) + R_3 (\WW (\omega)) $$ may be different, with
$$
          G^0_1 = G_1 \text{  and } G^0_2 = G_2 - \frac{1}{2} \partial^2_{\hat \eps \hat \eps} \Phi (0,\h) \  
$$
and 
$$
          | R_3 (\WW ) | \lesssim o( \hat\eps^2)  + \hat \eps^3 ||| \WW |||^3 \ . 
$$
(When $H<1/2$ we have $G^0_2 = G_2$ and 
the $o( \hat\eps^2)$-term can be written more quantitatively as $\hat\eps^2 \times O(\hat \eps^{1-2H})$.
In case $H=1/2$, we have $G^0_2 \ne G_2$, in turn the $o( \hat\eps^2)$-term is $O(\hat \eps^3)$.)

\medskip 

A LDP for rescaled lifted noise (in the space of models) is also given in \cite{bayer2019regularity} and thus induces Forde-Zhang type larges deviations for $\Phi^\eps ( \delta_{\hat \eps} \WW)$ with speed $\hat \eps^2$, without any growth assumptions on the volatility function $\sigma (.)$. The energy function $\I$ 
is smooth near zero \cite{bayer2017short}; \footnote{This is also a consequence of the abstract condition (A5) below, verified later on in the RoughVol example.} moreover  $\I(x) = \tfrac{1}{2} \| \h^x \|_\CM^2$ in terms of the unique non-degenerate minimizing control path $\h^x = (h^x, \bar h^x)$.

\subsection{Assumptions (A3-5): robust model specification and control theory} 
\label{sec:ass345}
Following the two running examples above -  classical diffusion / SDE based StochVol on the one hand and RoughVol on the other hand, we now present our general conditions. 
For the reader's convenience, the concrete regularity structure for RoughVol, upon which its discussion in Section \ref{sec:roughDEs} builds, is recalled in detail in Appendix \ref{app:ERS}. We have refrained, however, from further detailed recalls on the (by now well-known) rough path approach to SDEs. 
These examples, taken together, constitute our key examples of (asset price) models which admit a {\it robust representation}, in the sense of (A3a) and Appendix \ref{sec:Appendix3A}, in the form of a straight-forward list of properties, already highlighted in running examples, that are used in the sequel and should be no surprise to people familiar with rough paths and basic regularity structures, e.g. \cite{friz2014course}. 

\begin{itemize}
\item[(A3a)] 
There exists a regularity structure, with model space $\M$, that accommodates $m$-dimensional Brownian noise, and a continuous map $\Phi^\bullet: (\eps, \MM) \mapsto \Phi^\eps(\MM)$, from $[0,1]\times  \M \to C[0,1]$, such that we have the {\it robust representation},  
$$ \bar X^\eps (\omega) = 
\Phi^\eps ( \delta_{\beps} \WW (\omega)) \text{ a.s. } $$ 
where $\delta_{\beps} \WW$ is the $\beps$-dilation 
a random model $\WW = \WW (\omega)$ which arises from $m$-dimensional Brownian motion $\W = \W (\omega)$ as Wong-Zakai type limit, and on which 
$\CM$ acts by {\it translation} $T_\h, \h \in \CM$; Appendix \ref{sec:Appendix3A} gives more details.

\medskip

\item[(A3b)] A Schilder-type LDP holds {\it in model topology} for $ \delta_{\beps} \WW$ with (good) rate function given by
$ \mathbf{h} \mapsto \tfrac{1}{2} \| \h \|^2_{\H}$ when $\hh = (\h,...)$ is the canonical lift of $\h \in \CM$, and $\infty$ else. 
\end{itemize} 
\noindent 
Note that condition (A3), by application of the contraction principle, implies a LDP for $(\bar X^\varepsilon_1)$, and thus condition (A1a), with good rate function 
\begin{equation} \label{equ:inducedRF}
          \I (x) = \inf_{\h \in \H} \{ \tfrac{1}{2} \| \h \|^2_{\H}: \Phi _1( \h ) = x \} \equiv \inf_{\h \in \mathcal{K}^{x}}  \tfrac{1}{2} \| \h \|^2_{\H} \ ,
\end{equation}
where $\mathcal{K}^{x} \subset \CM$ denotes the space of $x$-{\it admissible controls}, i.e. elements $\h \in \H: \Phi_1 (\h) = x$
 where we abuse notation,  in terms of the canonical lift $\hh$ of $\h$, by writing
 $
    \Phi (\h) = \Phi^0 (\hh).
 $
%

\bigskip
\noindent

We also assume a robust ``stochastic" Taylor-like expansion, formulated however in a purely deterministically fashion. 
Care is necessary because in general $\eps \ne \beps$ and also
$$
    \Phi^\eps ( \delta_{\beps} \MM)) \ne \Phi^0 ( \delta_{\beps} \MM)
$$
where $ \delta_{\beps}$ denotes again the dilation operators on models; cf. Appendix \ref{sec:Appendix3A})
Let $\mathbb{B}$ be Banach space (typically $C[0,1]$ or $\R$). We call a map $G:  \M \mapsto \mathbb{B}$ {\it continuous linear} if there exists a continuous linear map $G_1: C^\alpha \to \mathbb{B}$ so that $G(\MM) = G_1(M)$. As a (trivial) consquence, $\sup \{ G (\MM) : ||| \MM ||| \le 1  \}  < \infty$ and, for all $\eps \ge 0$ and $\k \in \CM$,
$$ G (T_\k \delta_\varepsilon \MM) = G( \kk) + \varepsilon G (\MM) \ . $$
Similarly, call a continuous map $G: \MM \mapsto \mathbb{B}$ {\it continuous quadratic} if $\sup \{ G (\MM) : ||| \MM ||| \le 1  \}  < \infty$ and 
$$
                 G (T_\k \delta_\varepsilon \MM) = G ( \kk) + \varepsilon  G (\MM,\kk) + \varepsilon^2 G (\MM) 
$$
for a continuous bilinear map $G_2: C^{\alpha} \times \CM \to \mathbb{B}$ so that $G_2 (M, \k) =: G (\MM,\kk)$.
\begin{itemize}


\item[(A4a)] For every $\h \in \CM$ there exists $(G^{0,\h}_1,G^{0,\h}_2,R^{0,\h,\eps}_3)$ such that for every $\eps \ge 0$ and model $\MM = (M, ...) \in \M$ we have 
$$
\Phi^0_1 (T_\h \delta_{\beps} \MM ) = G_0^\h + \beps G^{0,\h}_1 (\MM) + \beps^2 G^{0,\h}_2 (\MM) + R^{0,\h,\eps}_{3} (\MM) \ . 
$$
with $ G^\h_0 = \Phi^0_1 ( \h )$, continuous linear $G^{0,\h}_1 : \MM \to  \R$, continuous quadratic  $G^{0,\h}_2: \MM \to  \R$, such that $(h , \MM) \mapsto G^{0,\h}_i(\MM)$ is continuous for $i=1,2$, and remainder estimate 
$$
           | R_3^{0,\h,\eps} (\MM) | \lesssim  o(\beps^2)  + \beps^3 ||| \MM |||^3 \ , 
$$
in the sense of a uniform estimate, over $\beps ||| \MM |||$ and $|| \h ||_{\CM} $ bounded. 
\item[(A4b)]  
Similarly, assume existence of $(G^\h_1,G^\h_2,R^{\h,\eps}_3)$ such that
$$
\Phi^\eps_1 (T_\h \delta_{\beps} \MM ) = G^\h_0 + \beps G^\h_1 (\MM) + \beps^2 G^\h_2 (\MM) + R^{\h,\eps}_3 (\MM) \ .
$$
with same remainder estimate and
$$
          G_1^{\h} (\MM) \equiv G^{0,\h}_1 (\MM)  \text{  and   }    G_2^{\h} (\MM) \equiv G^{0,\h}_2 (\MM) + K_2^{\h},
$$
for some $K_2^{\h} \in \R$. 
\end{itemize} 
%
%
\bigskip

\noindent  

Assumption (A4a) implies $C^2$ (Fr\'echet) differentiability of $\Phi$ on $\CM$. Indeed, taking $\MM=\mathbf{h}'$ be the canonical lift of an element $\h'$ of $\H$, (A4a) implies that 
\[ \Phi^0_1 (\h + \beps \h' ) = \Phi^0_1 (\h) + \beps G^\h_1 (\h') + \beps^2 G^\h_2 (\h') + o(\beps^2) \] 
for some continuous $G^\h_1, G^\h_2 : \H \to \R$ which are respectively linear and bilinear on $\H$. In particular, the (Fr\'echet) derivatives $D \Phi_1$ and $D^2 \Phi_1$ are well-defined. 

\medskip 

The next set of conditions is of control theoretic nature, with $x > 0$ fixed, and only concern  $\Phi = \Phi^0$ as map from $\CM$ to $C[0,1]$. Recall that $\mathcal{K}^{x} \subset \CM$ denotes the space of elements $\h \in \H: \Phi_1 (\h) = x$.

\begin{itemize}


\item[(A5a)] There exists a unique minimizer $\h^x \in \mathcal{K}^{x}$, with 
$$
 \I (x) =  \tfrac{1}{2} \| \h^x \|^2_{\H}. \ 
$$
\item[(A5b)] $\Phi_1: \CM \to \R$ has, at $\h^x$, a surjective differential, which here means
$$
                      D \Phi_1({\h^x}) \ne 0 \in \CM^*. 
$$

\item[(A5c)] The minimizer $\h^x$ is non-degenerate, in the sense that for all $0 \ne \h \in \H_0 :=\{D \Phi_1(\h^x) \}^{\perp}$, 
$$ q^x D^2\Phi_1(\h^x) (\h, \h) < \| \h \|^2_{\H},  $$ 
where $q^x \in \R$ is such that $\h^x= D\Phi_1(\h^x)^* q^x$.


\end{itemize}

\noindent
The existence of the Lagrange multiplier $q^x$ is discussed in Lemma \ref{lem:appendix-2pre}. Note that assumption (A5) is classical in the SDE context, cf. for instance \cite{azencott1985petites, kusuoka2008remark,deuschel2014marginala}. We discuss further Assumption (A5c) in Appendix \ref{app:SOO}, in particular it can be seen as strict positivity of the Hessian of $I(h):= \frac 1 2 \|\h\|^2$ when restricted to $\mathcal{K}^x$. Also note that in fact assumptions (A5a)-(A5c) imply that $\I$ is $C^1$ at $x$ and then one simply has $q_x = \I'(x)$, cf. Lemma \ref{lem:hx}.


\noindent 


\bigskip

We now assume (A4) and (A5). Fix $\h = \h^x$ as supplied by (A5) and the corresponding rough Taylor terms $(x, G^x_1,G^x_2,R^{x,\eps}_3)$
{$:=(G^{\h^x}_0, G^{\h^x}_1,G^{\h^x}_2,R^{\h^x,\eps}_3)$} supplied by Assumption (A4b). 
Let $\WW = \WW (\omega)$ be the It\^o lift of Brownian motion. We then define the (probabilistic) objects
\begin{equation} \label{eq:g}
           g_1^x (\omega) := G_1^{x} (\WW(\omega)), \ \ g_2^x (\omega) := G_2^{x} (\WW(\omega)), \ \  r^{x,\eps}_3 (\omega) := R^{{x},\eps}_3 (\WW(\omega)) \ .
\end{equation}
so that the stochastic Taylor expansion for the $\h^x$ Girsanov shift of $\bar X^{\eps}$ can be written as
$$
       \bar Z^{\eps, x} (\omega) =  x + \beps g_1^x (\omega) + \beps^2 g_2^x (\omega) + r^{x,\eps}_3 (\omega) \ .
$$

We saw in Section \ref{sec:SDEs} and \ref{sec:roughDEs} that condition (A4a-b) holds in the classical StochVol (SDE) case, as well as the RoughVol case. 

\section{Option pricing: exact representation formulae and asymptotics} 

\subsection{Call and put price formulae} 

We can now state a key formula which generalizes the RoughVol call price formula considered in \cite{bayer2017short} and which is at the heart of our analysis. It applies in the general setup of Section \ref{sec:ass345}. Therefore, it applies  
both to
classical StochVol situations (with $H=1/2$) where $\I = \mathcal{I}$ has the geometric interpretation of shortest square-distance to some arrival manifold determined by the strike, and 
to RoughVol with  the Forde-Zhang rate function $\I = \mathcal{J}$ as given in (\ref{ifz}).

\begin{theorem}\label{th:expansion}  (i) Assume (A1-A5) for fixed $x>0$. 
Let $k_\eps= x \eps / \beps $. 
Let $g_1^x, g_2^x$ and $r_3^{\eps,x}$ be the stochastic Taylor coefficients / remainder defined in \eqref{eq:g}.
Then 
$$
       c(\eps^2,k_\eps) =  \exp \left( {  - \frac{\I(x)}{\beps^2}  } \right) e^{k_\eps} J^{\mathrm{call}}(\eps, x)
       $$
with
\begin{equation}\label{defJfirst}
         J^{\mathrm{call}}(\eps, x) = E \left[ \exp \left( {  - \frac{\I'(x)g_1^x }{\beps}  } \right) \Big( \exp  \big(  \eps g_1^x + \beps \eps g_2^x +   ({\eps}/{\beps}) \, r_3^{\eps,x} \big) - 1\Big)^+\right] 
\end{equation}
 (ii) Assume now (A1),(A3-A5) for fixed $x<0$ and $k_\eps= x \eps / \beps $. 
Let $g_1^x, g_2^x$ and $r_3^{\eps,x}$ as before.
Then 
$$
       p(\eps^2,k_\eps) :=  \exp \left( {  - \frac{\I(x)}{\beps^2}  } \right) e^{k_\eps} J^{{\mathrm{put}}}(\eps, x)
       $$
with
\begin{equation}\label{defJput}
         J^{\mathrm{put}}(\eps, x) = E \left[ \exp \left( {  - \frac{\I'(x)g_1^x }{\beps}  } \right) \Big( 1- \exp  \big(  \eps g_1^x + \beps \eps g_2^x +   ({\eps}/{\beps}) \, r_3^{\eps,x} \big) \Big)^+\right] 
\end{equation}
(iii) 
Under Assumption (A1-A5), $\Lambda'(x) > 0$ for $x>0$ and under (A1) and (A3-A5), $\Lambda'(x) < 0$ for $x<0$. 
Writing $\sigma_x^2$ for the variance of $g_1^x$, the following differential relation holds, provided $x \ne 0$, under (A1-A5) for fixed $x>0$, and under (A1),(A3-A5) for fixed $x<0$: 
 \begin{equation}\label{eCool}
\sigma_x^2=\frac{2\I(x)}{\I'(x)^2} \ . 
\end{equation}
\end{theorem}

\begin{remark} Assumption (A2) in part (i) can be relaxed to integrability of $\exp \bar X^\eps_1$: one just needs finite call prices here. 
Recall furthermore (Lemma \ref{lem:hx}) that (A5) implies that $\Lambda$ is $C^1$ in a neighborhood of $x$.  
\end{remark} 
\begin{proof} By the very definition of call price function $c$ and $X_t \sim X_1^\eps$ with $t = \eps^2$,
$$
      c(t,k) =  E [ (\exp( X_1^\eps) - \exp (k)^+] .
$$
Since $X^\varepsilon := \bar  X^\varepsilon { \eps }/{\beps} $, consider $k=k_\eps$ and thus, with $1/ \nu_\eps =  { \eps }/{\beps} $,
$$
     c(\eps^2,k_\eps) = E [ (\exp( \bar X_1^\eps / \nu_\eps) - \exp (k_\eps))^+] .
$$
Our assumptions imply that there is a LDP for $ \bar X^\eps_1 (\omega) = \phi^\eps_1 (\beps \omega)$, where $ \phi^\eps$ denotes the It\^{o} map as in Section \ref{sec:SDEs}, such that 
$$
     - \eps^2 \log P [ \bar X_1^\eps > x ] \to \tfrac{1}{2} \| {\rm h}^x  \|^2_{\CM} = \I (x)
$$
for some unique minimizer of the control problem $\phi^0_1 ( {\rm h} ) = x$.
By Assumption (A4) we have the stochastic Taylor expansion of the form 
$$
       \bar Z^\eps_1 (\omega) = \phi^\eps_1 (\beps \omega + \h^x) =    x+ \beps g_1^x  + \beps^2 g_2^x   + r_3^{x,\eps}
$$
with the same notations as in \eqref{eq:g}.
Apply Girsanov's theorem, $\beps \W \to \beps \W + \h^x = \beps (\W + \h^x / \beps)$, and obtain
$$
          c(\eps^2,k_\eps) =  E [ e^{-\tfrac{1}{\bar \eps} \int_0^1 \dot{\h}^x d\W - \tfrac{1}{2\bar \eps^2} \| \h^x \|^2_{\CM }} (\exp( \bar Z_1^\eps / \nu_\eps) - \exp (k_\eps))^+] .
$$     
We then use first order optimality of $\h^x$ to see that (see Appendix \ref{app:FOO}) 
\begin{equation}\label{opt}
           \int_0^1 \dot{\h}^x d\W 
                          =\I'(x) g_1^x (\omega) \ 
\end{equation}
and this establishes the call price formula and hence part (i). With It\^o isometry we can write, from \eqref{opt},
\[
\|h^x\|^2_{\CM}
=\|\dot{h}^x\|^2_{L^2}=
(\I'(x))^2 Var(g_1^x)
\]
and conclude with
$\|h^x\|^2_{\CM}
=2\I(x)$. 
Recall from assumption (A1) that, for $x>0$, $-{\beps^2} \log P[ \bar X_1^\eps > x ] \sim  \I(x)$. By monotonicity of the left-hand side in $x$ one easily see that $\I(x)$ is monotone and hence $\Lambda'(x) \ge 0$. Moreover, we know from (A1b) that $\Lambda (x) = \| \h^x \|^2 / 2 \ne 0$.
In view of (\ref{opt}) we cannot have $\Lambda'(x) =0$ and so $\Lambda'(x) > 0$. This settles (iii) in case $x>0$. 

\medskip

\noindent The proof of (ii) and then (iii), in case $x<0$ is completely analogous.
 \end{proof}

\subsection{Precise large deviations}

\begin{theorem}\label{thm:main0} (i) Assume (A1-A5) for fixed $x>0$. 
Let $k_\eps= x \eps / \beps $. Then there exists a function $A=A(x) \sim 1$ as $x \downarrow0$, such that, with $\sigma_x^2=2\I(x) / \I'(x)^2$ as earlier, 
\begin{equation} \label{equ:main0}
       c(\eps^2,k_\eps) \sim \exp \left( {  - \frac{\I(x)}{\beps^{2}}  } \right)  \eps \, \beps^2 
       \frac{A(x)} 
       { (\I'(x))^2  \sigma_x \sqrt{2\pi}}\ \ \ \text{as $\eps \downarrow0$} .
\end{equation}
(ii) Now
assume (A1), (A3-A5) for fixed $x<0$. 
Let $k_\eps= x \eps / \beps $. Then there exists a function $A=A(x) \sim 1$ as $x \uparrow0$, such that
\begin{equation}
       p(\eps^2,k_\eps) \sim \exp \left( {  - \frac{\I(x)}{\beps^{2}}  } \right)  \eps \, \beps^2 
       \frac{A(x)} 
       { (\I'(x))^2  \sigma_x \sqrt{2\pi}}\ \ \ \text{as $\eps \downarrow0$} .
\end{equation}
\end{theorem}



We here derive the (correct) formula with a {\it formal} computation that ignores the remainder term. The actual proof is given in Section \ref{sec:proofmain} and relies in particular on Assumption (A4)  to handle the remainder. 

\begin{proof} In view of the exact call price formula in Theorem \ref{th:expansion}, it suffices to analyse
\begin{equation}\label{defJ}
         J(\eps, x) = E \left[ \exp \left( {  - \frac{\I'(x)g_1^x }{\beps}  } \right) \Big( \exp  \big(  \eps g_1^x + \beps \eps \, g_2^x 
         +  ({\eps}/{\beps}) \, r_3^{\eps,x} \big) - 1\Big)^+\right] .
\end{equation}
We ignore the remainder. With high probability, $(\eps g_1^x + \beps \eps g_2^x) $ is small when $\eps, \beps \to 0$, hence we expect
$$
        E \left[ \exp \left( {  - \frac{\I'(x)g_1^x }{\beps}  } \right) \Big( \exp  (  \eps g_1^x + \beps \eps g_2^x ) - 1\Big)^+\right] \sim  \eps \, E \left[ \exp \left( {  - \frac{\I'(x)g_1^x }{\beps}  } \right) \Big(   g_1^x + \beps g_2^x \Big)^+\right] \ .
$$
With $x$ fixed, write $g_i \equiv g_i^x$. Recall (first order optimality)
$$
      \I'(x) \,  g_1 (\omega) =  \int_0^1 \dot{\h}^x d\W 
$$
and we can then decompose 
$         g_2 =  \Delta _{2}+g_{1}\Delta_1 +g_{1}^{2}\Delta _{0} $
where $\Delta = (\Delta_0, \Delta_1, \Delta_2)$ is independent from $g_1$; 
see Lemma \ref{lem:AboutV} (alternatively, employ a Wiener-It\^o chaos argument). This leaves us with the computation of 
$$
        E \left[ \exp \left( {  - \frac{\I'(x)g_1 }{\beps}  } \right) \Big(  g_1+ \beps ( \Delta _{2}+g_{1}\Delta_1 +g_{1}^{2}\Delta _{0} )  \Big)^+\right] 
        = E \Big[ E \big( .... | \Delta \big) \Big] \ .
$$
The inner (conditional) expectation is a simple Gaussian integral for which a finite-dimensional Laplace analysis (cf. Lemma \ref{lem:BS1}) gives
$$
         E \big( .... | \Delta \big) \sim  \frac{\beps^2}{ (\I'(x))^2  \sigma_x \sqrt{2\pi}} \exp \big( \I'(x) \Delta_2 \big) \  \ \text{ as $\beps \to 0$.}
$$   
The asymptotic behaviour of the full expectation is then indeed obtained, as one would hope, by averaging over $\Delta_2 = \Delta_2^x (\omega)$, so that
$$
        E \left[ \exp \left( {  - \frac{\I'(x)g_1 }{\beps}  } \right) \Big(   g_1+ \beps ( \Delta _{2}+g_{1}\Delta_1 +g_{1}^{2}\Delta _{0} )  \Big)^+\right] 
        \sim  \frac{\beps^2}{ (\I'(x))^2  \sigma_x \sqrt{2\pi}} E \Big[ \exp \big( \I'(x) \Delta_2^x \big) \Big] \ .
$$
Clearly, such a formula requires $\exp \big( \I'(x) \Delta_2^x \big) \in L^1(P)$; in fact, the proof requires $L^{1+}(P)$ and 
we see in Section \ref{sec:LocAna} that this is precisely the case because $\h^x$ is a {\it non-degenerate} minimizer, cf. assumption (A5). Taking into account the factor $e^{k_\eps}=e^{x \eps / \beps}$ from Theorem \ref{th:expansion}, see that 
\begin{equation} \label{equ:Aabstract}
 A(x) = 
\begin{cases}
    E \big[ \exp ( \I'(x) \Delta_2^x ) \big],& \text{if } H < 1/2\\
    e^x E \big[ \exp ( \I'(x) \Delta_2^x ) \big],& \text{if } H=1/2
\end{cases}
\end{equation}
The formula in the put case follows from the computation of 
$$
        E \left[ \exp \left( {  - \frac{\I'(x)g_1 }{\beps}  } \right) \Big(   -g_1- \beps ( \Delta _{2}+g_{1}\Delta_1 +g_{1}^{2}\Delta _{0} )  \Big)^+\right] 
$$
that is exactly as before once we have written $ \I'(x)g_1= (-\I'(x))(-g_1)$ in the exponential where we recall $sgn(\I'(x))=sgn(x)$. 
\end{proof}

\begin{remark}
Sanity check: Black-Scholes $H=1/2$, with $\sigma^x \equiv \sigma$, $\I(x) = x^2 / (2 \sigma^2)$ . We then have 
$$
      c(\eps^2,x) \sim  \exp \left( {  - \frac{x^2}{2 \eps^2 \sigma^2 }  } \right) \varepsilon^{3} \frac{\sigma^3 }{ x^2 \sqrt{2\pi}} A(x) 
$$
which matches precisely the previously derived Black-Scholes expansion (\ref{e:BScall}), with 
$ A(x) = e^{x (1+\mu / \sigma^2)} $ as predicated by (\ref{equ:Aabstract}) with $\Delta_2 = \mu$. Remark that in the Black-Scholes case, assumptions (A1-A5) are indeed satisfied for any $x>0$.

\end{remark}

\subsection{Precise moderate deviations}

We now turn to the moderate regime. The proof of the following result is given in Section \ref{proof:MD}. We only spell out the call case here,
leaving the similar case of puts to the reader.

\begin{theorem} \label{thm:mdp}
Assume that Assumptions (A1)-(A5) hold for $x=0$ with $h^0=0$.
Let $k_\varepsilon = x_\varepsilon \varepsilon/ \bar{\varepsilon}$ with $x_\varepsilon \to 0$, $x_\varepsilon/\bar{\varepsilon} \to \infty$. Then

\begin{equation} \label{equ:modCexp}
c(\varepsilon^2,k_\varepsilon) \sim_{\varepsilon \to 0} \exp\left(- \frac{ \Lambda(x_\varepsilon)}{\bar{\varepsilon}^2} \right) \varepsilon \bar{\varepsilon}^2 \frac{ \sigma_0^3}{x_\eps^2 \sqrt{2\pi} }.
\end{equation}
\end{theorem}


%
%
%
%
%
%
%
%
%
%
\begin{remark}
Formally, this follows from our precise large deviations (\ref{equ:main0}) by replacing $x$ by $x_\varepsilon$, and using $x_\varepsilon \to 0$. The rigorous proof follows along similar lines as the proof of Theorem \ref{thm:main0} and is postponed to Section \ref{sec:proofmain}.
\end{remark}
\begin{remark} Let $\bar\eps = \eps^{2H}$ and consider $\beta \in (2H/3,H)$. Then $x_\eps = x \eps^{2\beta}$ falls in the  regime of the above theorem and moreover,
$$
     \I (x_\eps) / \eps^{4H} \sim \frac{x^2}{2 \sigma_0^2 \eps^{4H-4\beta}} \ \ \ \text{ as $\eps \to 0$} 
$$
and in fact the expansion in (\ref{equ:modCexp}) is nothing else than the Black-Scholes expansion run in the moderate scale, with speed function $\eps^{4H-4\beta}$, instead of the large deviation speed $\eps^2$. 
More generally, one can obtain for arbitrary $\beta \in (0,H)$ the expansion
$$
     \I (x_\eps) / \eps^{4H} = \sum_{k=2}^M \frac{\Lambda^{(k)}(0)}{k !}   \frac{x^k}{\eps^{4H-2 k \beta }  }+ o(1) \ \ \ \text{ as $\eps \to 0$} 
$$

where $M$ is such that $(M +1) \beta > 2H$, if we know that $\Lambda$ is $C^{M}$ at $0$. (Note that under assumption (A5), $C^M$ regularity of $\Lambda$ at $x$ simply requires $C^{M}$ regularity of $\Phi$ on $\H$, cf Lemma \ref{lem:hx}.)
\end{remark}

\section{Return to RoughVol} 


\subsection{Checking the abstract conditions}

We return to RoughVol model, Example \ref{ExRVLDP},
\begin{equation*}
X_1^{\varepsilon }=\int_{0}^1 \s \left( \hat \eps \hat W \right) \varepsilon
d\left( \bar{\rho} \bar W+\rho W\right) -\frac{1}{2}\varepsilon ^{2}\int_{0}^1
\s^2 \left(  \hat \eps \hat W\right) dt \ ,
\end{equation*}
with smooth volatility function $\s (.)$, and Forde--Zhang energy function $\J$  as given in (\ref{ifz}). Write $\s_0 = \s (0) >0 $ for spot-vol and also set $\s'_0 = \s' (0)$.  

\medskip 

Applied to this model, Theorem \ref{thm:main0},  yields the following result.

\begin{corollary}[RoughVol] \label{thm:main44} Let $H \in (0,1/2]$ and $k_\eps = x  \eps^{1-2H} > 0$. Assume that assumption (A2) is satisfied. Then, for $x$ small enough, $\J = \J(x)$ is continuously differentiable and
\begin{equation}
       c(\eps^2,k_\eps) \sim \exp \left( {  - \frac{\J(x)}{\eps^{4H}}  } \right)  \varepsilon^{1+4H} \frac{A(x)}{ (\J'(x))^2  \sigma_x \sqrt{2\pi}} \ \ \text{as $\eps \downarrow0$,}
\end{equation} 
for some function $A(x)$ with $A(x) \to 1$ as $x \to 0$. A similar expansion holds, without assumption (A2), for out-of-the-money ($x<0$) put prices.
\end{corollary}

%

\begin{remark}
It is known that Assumption (A2) holds when $\sigma$ has linear growth, cf. \cite{forde2017asymptotics}. In the case $H=1/2$, (A2) holds under much weaker assumptions, e.g. for $\sigma$ of exponential growth and correlation $\rho <0$ \cite{sin1998complications, jourdain2004loss, lions2007correlations}. We expect similar results to hold in the rough regime but they are not known at the moment \cite{gassiat2018martingale}.
\end{remark}

\begin{proof}

It suffices to check that all the assumptions of the theorem are satisfied for $x$ small enough. As mentioned in Section \ref{sec:uniLDP}, Assumption (A1) follows from \cite{forde2017asymptotics,bayer2019regularity}. Assumptions (A3) and (A4) (i.e. the regularity structure framework) follow from \cite{bayer2019regularity}, which we already reviewed in Section \ref{sec:roughDEs}, see also  Appendix \ref{app:ERS} for the reader's convenience. At last, we discuss Assumption (A5). We first check that it holds at $x=0$. (A5a) is obvious with $h^0=0$. By a direct computation one has, writing $ \tilde{k}=\rho k+\bar{\rho}\bar{k}$ and $\mathrm{k} = (k, \bar{k})$,
\begin{equation}   \label{checkRoughVol1}
 \left\langle D\Phi_1(0), \mathrm{k}  \right\rangle= \sigma_0 \tilde{k}_1, 
 \end{equation}
so that strictly positive spot-vol $\sigma_0$ implies (A5b). Finally, since $q^0=0$, (A5c) is trivial.
Now, Lemma \ref{lem:hx} (combined with Lemma \ref{lem:WeakCont}) shows that conditions (A5a-c) are ``open'', i.e. automatically holds in a small neighbourhood of $0$. This concludes the proof.

\end{proof}

\subsection{Computing the constant}

The merit of the abstract framework has been shown by applying it to the concrete case of RoughVol, as detailed above. It (still) follows from the abstract framework that
\begin{equation} 
 A(x) = 
\begin{cases}
    E \big[ \exp ( \I'(x) \Delta_2^x ) \big],& \text{if } H < 1/2\\
    e^x E \big[ \exp ( \I'(x) \Delta_2^x ) \big],& \text{if } H=1/2
\end{cases}
\end{equation}
{but we can now be more specific about the random variable $\Delta_2^x$ (following the proof in next Section \ref{sec:proofmain}). Indeed, a straight-forward computation using
\eqref{def:g2}, \eqref{def:g2rough}
and \eqref{G2W}, \eqref{G2V}
shows that 
\begin{equation} \label{equ:DeltaForRV}
\Delta_2^x =\begin{cases}
 \frac{1}{2}\int_0^1 \sigma''(\hat{h}_s^x) \hat{V}^2_s d\tilde{h}_s^x + \int_0^1\sigma'(\hat{h}_s^x) \hat{V}_s d \tilde{V}_s, & \mbox{ if } H<1/2 \\
 \frac{1}{2}\int_0^1 \sigma''(\hat{h}_s^x) \hat{V}^2_s d\tilde{h}_s^x + \int_0^1\sigma'(\hat{h}_s^x) \hat{V}_s d \tilde{V}_s
-\frac{1}{2} \int_0^1 \sigma^2(\hat{h}_s^x) ds
, & \mbox{ if } H= 1/2.
\end{cases}
\end{equation}
Here, $\hat{V}$, $\tilde{V}$ can be computed as follows. We have that 
\[
g_1^x 
=
 \int_0^1\sigma \big(\, \hat h \, \big) d\tilde W+\int_0^1\sigma ^{\prime }\big(\, \hat h \, \big) \hat W
d \tilde h
\]
is the first order term in
\eqref{equ:dotsCaseHonehalf}. Then,  following \eqref{eq:defv}, we define $v_t=  {E} \left[ W_t g_1^x(W, \bar W) \right]/ E\left[(g_1^x)^2\right]$ and $\bar v_t=  {E} \left[ \bar W_t g_1^x(W, \bar W) \right]/E\left[(g_1^x)^2\right]$.
Then, $\hat{V}$, $\tilde{V}$ are obtained from $(V_t, \bar{V}_t) := (W_t-v_t g^x_1, \bar{W}_t - \bar{v}_tg^x_1)$. With such construction, $(V,\bar V)$ is independent from the Gaussian $g_1^x$. This follows from the (general) discussion below and shows how these abstract results can be 
implemented in practice. We refer to \cite{friz2019explicit} for further computations and numerical tests.}

\subsection{Remarks on RoughBergomi with $t^{2H}$ dependence and Volterra dynamics} 
\label{sec:rmkRB}

The key aspect of RoughBergomi model, following \cite{bayer2016pricing}, is its log-normal volatility, naturally modelled by the volatility function $\sigma (x) = \sigma_0 \exp (\eta x)$. Strictly speaking, in \cite{bayer2016pricing} the Wick -, rather than standard, exponential was employed which amounts to an additional factor of the form $\exp(-c t^{2H})$, or $\exp(-c\hat{\eps}^2 t^{2H})$ after rescaling. Let us indicate the required adaptations in the generality of a (non-growth restricted, sufficiently) smooth function $\sigma = \sigma(x,\tau)$. First
\[ \bar{X}_1^{\eps} = \int_0^1 \sigma (\hat{\eps} \hat{W}_t,
   {\hat{\eps}^2 t^{2 H}}) \hat{\eps} d (\rho W_t +
   \bar{\rho} \overline{W_t}) {- \frac{1}{2} \eps
   \hat{\eps} \int_0^1 \sigma^2 (\hat{\eps} \hat{W}_t,
   \hat{\eps}^2 t^{2 H}) d t} , \]
arises in the same way as the time-independent RoughVol from short time to small noise conversion and subsequent rescaling.  
From an LDP perspective, one suspects (correctly) that one can instead consider 
\[ \int_0^1 \sigma (\hat{\eps} \hat{W}_t, 0) \hat{\eps} d (\rho
   W_t + \bar{\rho} \overline{W_t}) = : \Phi^0 (\delta_{\hat{\eps}} \mathbf{W}).
\]
To see this it suffices to note a LDP for the pair 
 $\left(\delta_{\hat{\eps}} \mathbf{W}, \hat{\eps} \right)$ with good rate function
$ I (\mathbf{h}, \tau) = \frac{1}{2} \| h, \bar h \|^2_{\H}  $
if $\mathbf{h}$ is the canonical lift of
$(h, \bar{h})$  {\it and} $\tau = 0$, otherwise it takes value $+\infty$. 
%
%
%
Using joint continuity of
\[ (\mathbf{W}, {\tau}) \mapsto \int_0^1 \sigma (\hat{W}_t,
   {\tau^2} t^{2 H}
   ) d (\rho W_t + \bar{\rho} \overline{W_t})
   {- \frac{1}{2} {\tau^{1 + 1 / (2 H)}} \int_0^1
   \sigma^2 (\hat{W}_t, {\tau^2} t^{2 H}) d t} \]
the contraction principle then readily implies that $\bar{X}_1^{\eps}$ satisfies an LDP with good rate function (\ref{ifz}),
as does $\Phi^0 (\hat{\eps} \mathbf{W})$. 
%
The expert reader may notice that $t^{2 H}$ is not controlled by $\left(W_t,
\bar{W}_t\right)$, so that the claimed continuity requires justification beyond
(standard) rough integration / reconstruction. This is precisely where {\it singular modelled distributions} \cite{hairer2014theory}
come in, revisited from a rough path and rough vol perspective in \cite{bellingeri2020singular}. In essence, $t^{2H}$ is smooth (as {\it singular} modelled distribution) with singularity parameter $\eta = 2H$, so that (using stability under composition with regular maps) the integrand $t \mapsto \sigma (\hat{W}_t, {\tau^2} t^{2 H})$ and then the product with $\rho \partial_t {W}_t + \bar{\rho} \partial_t \overline{W_t}$ can be viewed as singular modelled distribution. This restores continuity and so Assumption A3a is indeed valid, as is (unaffected by all this) A3b; large deviations then become a consequence of the contraction principle. The real interest here is of course the robust specification which allows for the local expansions, as requested in Assumption A4. (Assumption A5 is also unaffected by all this.) In summary, at the only price of an extended justification for A3a, the RoughBergomi model with singular time dependence fits perfectly in the abstract framework laid out in this paper.

\medskip

A second remark concerns more complicated {\it Volterra volatility dynamics}. In both Forde-Zhang \cite{forde2017asymptotics} and rough Bergomi \cite{bayer2016pricing}, rough volatility is specified by an explicit expression involving fractional Brownian, no differential equation has to be solved. Such ``rough'' volatility models were dubbed ``simple'' in \cite{bayer2019regularity}. This is in contrast to ``non-simple'' models that require solving a Volterra type It\^o SDEs, as in the case
of roughHeston. Disregarding square-root (Heston) situations (but see Remark I.(ix) in the introduction) a robust solution theory for Volterra models is possible and requires a more involved regularity structure when $H < 1/4$ than what we have displayed here. We refer to \cite[Ch.5]{bayer2019regularity}, noting that Volterra SDEs fall in the general solution theory in Hairer's theory, so that solutions are found by fix point arguments in suitable spaces of modelled distributions. Since this solution theory effectively takes place in the space of local expansions - the essence of controlled rough paths and modelled distributions -  the robustness condition (A3a) and (A4) will remain valid in this setting, as does (A3b), relying on the general Gaussian model large deviations already employed in \cite{bayer2019regularity}. Finally, checking condition (A5) will again rely on continuity arguments and again require $x$ small. (The absence of an explicit expression for $\Phi_1$, which came in handy in (\ref{checkRoughVol1}),

\section{Proof of main result} \label{sec:proofmain}

We complete the proof of Theorem \ref{thm:main0} 

\subsection{Localization of $J$}
The expressions $J^{\mathrm{call}}$ and $J^{\mathrm{put}}$ are introduced in (\ref{defJfirst}) and (\ref{defJput}), respectively.
In what follows we only treat $J=J^{\mathrm{call}}$. The case of $J^{\mathrm{put}}$ is similar but easier (since the factor $(1-\exp (...))^+$ which appears in $J^{\mathrm{put}}$ stays bounded.) 

\medskip 

We first introduce a localized version of $J=J^{\mathrm{call}}$ as given in (\ref{defJfirst}), That is, we set 
\begin{equation}\label{defJdelta}
         J_\delta (\eps, x) = E_\delta \left[ \exp \left( {  - \frac{\I'(x)g^x_1 }{\beps}  } \right) (e^{ (\beps / \nu_\eps) g_1^x + (\beps^2/ \nu_\eps) g_2^x + (1/ \nu_\eps)  r_3^{\eps,x}} - 1)^+\right] \ .
\end{equation}
where the expectation is with respect to the sub-probability
$$
            P_\delta (A) = P ( A \cap \{  \beps ||| \WW ||| < \delta     \} ) \ .
$$


\begin{proposition} \label{prop:Jloc}
Fix $\delta > 0$. Then there exists $c=c_{x,\delta}>0$ such that
$$
     | J_\delta (\eps, x) - J (\eps, x) | = O(\exp (- c / \bar \eps^2)).
$$
Hence any ``algebraic expansion'' of $J$ (in powers of $\beps$) is unaffected by switching to $J_\delta$.
\end{proposition}

\begin{proof}
We revert to the respective expression of $J$ and $J_\delta$ before the Girsanov shift. To this end, introduce 
$$
  B := \{   \MM \in \M: ||| T_{-\h^x} 
   \MM ||| \ge \delta     \} \ ,
$$
noting that $B^c$ is a neighbourhood (in model topology) of the canonical lift of $\h^x$ by 
continuity of the translation operator.
This allows to write, using \eqref{defJ}, \eqref{defJdelta} and the Girsanov transform in Theorem \ref{th:expansion}, we have
\begin{equation}\label{bdloc}
J (\eps, x) -J_\delta (\eps, x)
= \exp \left( {  \frac{\I(x)}{\bar{\eps}^2}  } \right) e^{-k_\eps}
E [ (\exp( X_1^\eps) - \exp (k_\eps))^+ 1_B(\delta_{\beps} \WW) ] > 0   
\end{equation}
where $\h^x$ is (by assumption) the {\it unique} minimizer. We only need to upper bound this expression, since it is always positive. Let us write
\[
E [ (\exp( X_1^\eps) - \exp (k_\eps))^+ 1_B(\delta_{\beps} \WW) ] = 
E [ (\exp( \bar{X}_1^\eps/\nu_\eps) - \exp (x/\nu_\eps))^+ 1_B(\delta_{\beps} \WW) ] 
\]
We localize the call payoff $\psi(z) := (e^z - 1)^+$ to an ATM neighbourhood. 
For fixed $b>x$ we split the expectation over the two sets $\{\bar{X}_1^\eps \ge b\}$ and $\{\bar{X}_1^\eps<b\}$. 
 We have
\[
\begin{split}
&E[ (\exp( \bar{X}_1^\eps/\nu_\eps) - \exp (x/\nu_\eps))^+ 1_B(\delta_{\beps} \WW) ] \\ 
& \leq
E(\exp( \bar{X}_1^\eps/\nu_\eps) - \exp (b/\nu_\eps))^+ 
+
E [(\exp(b/\nu_\eps) - \exp (x/\nu_\eps))^+1_{\bar{X}_1^\eps\geq b}] 
\\
&+ E [(\exp( \bar{X}_1^\eps/\nu_\eps) -  \exp (x/\nu_\eps))^+1_{\bar{X}_1^\eps<b} 1_B(\delta_{\beps} \WW) ] 
\end{split}
\]
Because of the call price (upper) large deviation estimate in Proposition \ref{prop:callLDP}
\begin{equation}\label{LDP1}
E[(\exp( \bar{X}_1^\eps/\nu_\eps) - \exp (b/\nu_\eps))^+] \le \exp(-(\I(b)+o(1))/\beps^2)
\end{equation}
and since $b>x$ it is clear that  $\I(b)>\I(x)$. The same is true for
\begin{equation}\label{LDP2}
E [(\exp(b/\nu_\eps) - \exp (x/\nu_\eps))^+1_{\bar{X}_1^\eps>b}] 
\leq 
\frac{(\exp(b) - \exp (x))^+}{\nu_\eps} P(\bar{X}_1^\eps\geq b) \, 
\end{equation}
since $\nu_\eps>1$.
It remains to deal with the localized term, using $\nu_\eps \geq 1$ for $\eps \leq 1$,
\begin{equation} \label{equ:expbarx}
E [(\exp( \bar{X}_1^\eps/\nu_\eps) -  \exp (x/\nu_\eps))^+1_{\bar{X}_1^\eps<b} 1_B(\delta_{\beps} \WW) ] 
\le
 e^b  P [ \{ \bar{X}_1^\eps  \in [x,b) \} \cap B ] \ . 
\end{equation}
%
%
An upper bound on this is given by $e^b$ times
\[
\begin{split}
P [ \{ \bar{X}_1^\eps\geq x \} \cap B ] 
&=
P[ \Phi^\eps_1 ( \delta_{\beps} \WW ) \geq x;  ||| T_{-\h^x} \delta_{\beps} \WW ||| \ge \delta  ]\\
\end{split}
\]
Introduce the set 
\[
A^{x,\delta}= \{ \MM :  \Phi^0_1 ( \MM) \geq x;  ||| T_{-\h^x} \MM ||| \ge \delta \} 
\]
By assumption (A3) we find that
\[
P[ \Phi^\eps_1 ( \delta_{\beps} \WW) \geq x;  ||| T_{-\h^x} \delta_{\beps} \WW ||| \ge \delta  ] 
%
\le e^{- (\| {\h}_{x,\delta} \|^2_{\CM} + o(1))/(2 \beps^2)}  \ .
\]
where (as before, $\hh$ denotes the canonical lift of $\h \in \CM$ to a model) 
$$
          \h^{x,\delta}  \in {\rm argmin } \inf_{\h \in \CM} \{ \tfrac{1}{2} \| \h \|^2_{\CM}: \hh \in A^{x,\delta} \} 
$$
(The infimum over a closed set of a good rate function is attained, although there may be many minimizers.) 
Since { $\Phi_1(\h^{x,\delta}) \geq x$, by monotonicity of the rate function on $(0,+\infty)$}, we have $\| \h^{x,\delta} \|_{\CM} \ge \| \h^{x} \|_{\CM}$. But this inequality must be strict, for otherwise the assumed uniqueness of the minimizer implies $$\h^{x,\delta} = \h^{x}$$ which is not possible since $\h^{x,\delta} \in A^{x,\delta} \subset B$, i.e. outside a neighbourhood of $\h^x$. Set $2 \eta := \| \h^{x,\delta} \|_{\CM}^2 -\| \h^{x} \|_{\CM}^2$.
The proof is then finished by setting 
%
$c = \min \{ (\I (b) - \I (x))/2, \eta / 2   \} > 0.$
\end{proof}

%
%

\subsection{Local analysis} \label{sec:LocAna}

With $\h^x \in \CM$ fixed, assumption (A4) provides us with $(G_1^x,G_2^x, R_3^{x,\eps})$ so that, restricted to $\CM$,
$$
            G^x_1 = D\Phi_1(\h^x), \ \ G_2^x = D^2\Phi_1(\h^x) \text{ (as quadratic form)} \ .
$$
On the other hand, with $K_2$ as in Assumption (A4b),
\begin{equation}\label{G2W}
             g^x_1 (\omega ) = G^x_1 (\WW (\omega)), \    g^x_2 (\omega ) = G^x_2 (\WW (\omega)) = G^{0,x}_2 (\WW (\omega)) + K_2
\end{equation}
Clearly $g^x_1$ is a zero mean Gaussian, $N (0, \sigma^2_x)$, say. We then proceed as in \cite{azencott1985petites} and introduce the zero mean Gaussian process $\V =\V^x$ 
\begin{equation} \label{defV}
       \V_t (\omega)  := \W_t (\omega) - g^x_1(\omega) \v_t
\end{equation}
where $\v$ is chosen so that $\V$ is independent from $g^x_1$. 
We now describe such a $\v$ explicitly, which also requires identifying $g_1^x$ as Wiener integral.

\begin{lemma} \label{lem:gv}
(i) Identifying $D\Phi_1(\h^x)$  with an element of $\H$, one has the equality of random variables
\begin{equation} \label{eq:wiener}
g_1^x(\omega) =  \left\langle D\Phi_1(\h^x) , \W \right\rangle
\end{equation}
where $\left\langle \h, \W \right\rangle \equiv \int_0^1 \dot{\h} d \W$ denotes the Wiener integral.  \\
(ii) Define $\v = \v^x$ by  
\begin{equation}\label{eq:defv}
\v = \frac{D\Phi_1(\h^x)}{\| D\Phi_1(\h^x) \|^2}, \;\;\;\; \V(\omega) = \W(\omega) - \v g_1^x(\omega) .
\end{equation}
Then $\V$ is independent from $g_1^x$. \\
(iii) The Cameron--Martin space of $V$ is given by $\TH :=  \{ D\Phi_1(\h^x) \}^\perp$. 
\end{lemma}
\begin{proof}

(i) By definition, for all $\k $ in $\H$ one has $ \left\langle D\Phi_1(\h^x), \k \right\rangle_{\H} = G_1^x(\mathbf{k})$. Let $\W^\eta$ a mollifier function, with rescale parameter $\eta>0$ and call $\WW^\eta$ the canonical model lift of $ \W^\eta$. { By assumption (A3) there exists a renormalized approximation $\hat{\WW}^\eta =\mathfrak{R}^\eta \WW^\eta$ which converges as $\eta \to 0$, in probability and model topology, to $\WW$.  Hence we have by assumption (A4)}
\begin{align*}
& G_1^x(\WW) = \lim_{\eta \to 0} G_1^x( \mathfrak{R}^\eta \WW^\eta) = \lim_{\eta \to 0} G_1^x( W^\eta) \\
&= \lim_{\eta \to 0} \left\langle D\Phi_1(\h^x) , \W^\eta \right\rangle =  \left\langle D\Phi_1(\h^x) , \W \right\rangle.
 \end{align*}
(ii) It suffices to show that for all $\k \in \H$, $E\left[ \left\langle \V, \k \right\rangle g_1^x \right]= 0$, which is an easy consequence of It\^o isometry and the definition of $\v$. \\
(iii) If $D\Phi_1(\h^x)=0$, we have $\V=\W$ and there is nothing to show. Otherwise, normalize $D\Phi_1(\h^x) \in H$ to norm one, and complete to an orthonormal basis, say $e^{(n)}$ of $H$. By standard arguments, recover Brownian motion $\W = \sum_n \left\langle e^{(n)}, \W \right\rangle e^{(n)} $, in probability and uniformly on $[0,1]$, as Gaussian process with Cameron--Martin space $\H$. { The
statement follows by noting that $V$ is also given by this series with first mode, proportional to $\h^x$, removed. }
\end{proof}

Note that $\V$ is not adapted to the filtration generated by $\W$ so that a random model which lifts $\V$ cannot be constructed by elementary It\^o-integration. Although $\VV$ could then be constructed as (renormalized) Gaussian model, cf. \cite[Sec. 10.2]{hairer2014theory}), it seems more self-contained to lift  (\ref{defV}), i.e.
$$
            \VV (\omega) := T_{- G^x_1(\WW(\omega)) \v} \WW (\omega) \ .
$$
\begin{lemma} \label{lem:AboutV}
 (i) The Gaussian (random) model  $\VV = \VV (\omega)$ is independent, in the sense of random variables, of  $g^x_1 = g^x_1 (\omega)$.

\noindent (ii) There exists a  $\Delta^x = (\Delta^x_0, \Delta^x_1, \Delta_2^x)$, independent from $g^x_1 (\omega)$, such that (omit superscript $x$)
$$
     g_2 =  \Delta _{2}+g_{1}\Delta_1 +g_{1}^{2}\Delta _{0}
$$

\noindent (iii) The rescaled family $\delta_\eps \VV$ satisfies a LDP in model topology, with good rate function given
by
$$
                 {\vv} \mapsto \tfrac{1}{2} \| \h \|^2_{\H}   \text{ whenever $\vv$ is the canonical lift of $\h \in \TH$, and $+\infty$ else.}
$$
\noindent (iv) $\exp(||| \VV |||^2) \in L^{0+}$.
\end{lemma}
\begin{proof} (i) As in the previous proof let $\W^\eta$, $\V^\eta$ denote convolution with a mollifier function, with rescale parameter $\eta>0$. Of course, $\V^\eta \to \V$ uniformly with uniform H\"older bounds. Using the same notation for mollification of $\W$ and $\v$, have
$$
   \V^\eta = \W^\eta - g^x_1(\omega) \v^\eta \ .
$$
Recall that by (A3), one has convergence of a renormalization of the canonical lift  $\hat{\WW}^\eta =\mathfrak{R}^\eta \WW^\eta$  as $\eta \to 0$, in probability and model topology, to the It\^o-model $\WW$.  Since translation commutes with renormalisation 
$$
       \mathfrak{R}^\eta T_{- G^x_1(\WW) \v^\eta} \WW^\eta = T_{- G^x_1(\WW) \v^\eta} \hat \WW^\eta \to T_{- G^x_1(\WW)\v} \WW =  \VV \text{ \ as $\eta \to 0$ \ .}
$$
It now suffices to note that $T_{- G^x_1(\WW) \v^\eta }(\WW^\eta)$ is precisely the canonical model lift of, and hence a measurable function of $\V^\eta = \W^\eta - g^x_1 (\omega) \v^\eta = \W^\eta - G^x_1 (\WW) \v^\eta$.

\noindent
(ii) Straightforward from the definition $g_2^x (\omega) = G_2 (\WW (\omega)) = G^0_2 (\WW (\omega)) + K_2$. By assumption (A4), the map $G^0_2$ is continuous quadratic; applied with $\WW = T_{g_1(\omega) \v} \VV$ this gives the claimed decomposition of $g_2$, with
\begin{equation}\label{G2V}
          \Delta_2 (\omega) =  G_2^0 (\VV (\omega)) + K_2, \ \Delta_1 (\omega) =  G_2^0 ( \VV (\omega), \v ) \text{ and } \ \Delta_0  =  G_2^0 (\v) \ .
\end{equation}
\noindent (iii) True from general facts for Gaussian models. Alternatively, recall that  $G_1^x(\WW) =\left\langle D\Phi_1(\h^x), W\right\rangle$, 
so that for $\h$ in $\H$, by the definition \eqref{eq:defv} of $\v$,
$$\h - G_1^x(\mathbf{h}) \v = \h - \frac{\left\langle D\Phi_1(\h^x), \h \right\rangle}{ \|D\Phi_1(\h^x)\|^2} 
D\Phi_1(\h^x)
= P_{\TH} h, $$
with  $P_{\H_0}$ the orthogonal projection on $\H_0$. %
 Recalling the LDP satisfied by $\WW$, by the contraction principle this implies that $\delta_{\varepsilon}  \mathbf{V} = T_{-G_1^x( \delta_\eps \WW) \v} \delta_\eps\mathbf{W}$  satisfies a LDP with rate function given by
$$I(\Pi) = \inf\{\frac{1}{2} \|h\|^2, \;\; h \in \TH, \Pi = \mathbf{h} \}.$$

(iv) Again, true from general facts for Gaussian models for which Fernique estimates are available. Alternatively, only use that 
  $\WW$ is a Gaussian model and we have a  Fernique estimate for its homogenous norm $||| \WW |||$, cf. {\eqref{eq:Fernique}}. Since  $|||\VV||| \lesssim |||\WW|||+\|\v^x\| |g_1|$ where $g_1$ is a Gaussian, the claim follows.
\end{proof} 

We finally show that the non-degeneracy assumption (A5c) is actually equivalent to an exponential integrability property for the Wiener functional $\Delta_2^x$ (defined in Lemma \ref{lem:AboutV} (ii) above). Recall that (A5) implies $C^1$ regularity of $\I$ in a neighbourhood of $x$.

\begin{lemma} \label{lem:NDstrong}
Under assumptions (A3)-(A4), the non-degeneracy (A5c) is equivalent to the seemingly stronger assumption that
\begin{equation}
\exists \beta < 1: \;\;\; \forall \h\in\H,  q_x D^2\Phi_1(\h^x) (\h, \h) \leq \beta \| \h \|^2_{\H}.
\end{equation}
\end{lemma}

\begin{proof}
It suffices to prove that if (A5c) holds, then
\[ \beta := \sup \left\{   q_x D^2\Phi_1(\h^x) (\h, \h), \;\;\; \| \h \|_{\H} \leq 1 \right\} < 1. \]
But assumption (A4) implies that 
\[D^2\Phi_1(\h^x) (\h, \h) = G_2(\hh) \]
where $G_2$ is continuous on $\M$. Now recall that as part of Assumption (3b), we assume that $ \hh \mapsto \frac 1 2 \| \h \|_{\H}^2$ is a {\it good} rate function, so that the canonical lift maps bounded sets in $\H$ to (relatively) compact subsets of $\M$. It follows that the supremum in the definition of $\beta$ is attained at a $\h^\ast$, which in turn implies by (A5c) that
$$\beta = q_x D^2\Phi_1(\h^x) (\h^\ast, \h^\ast) < \| \h^\ast \|^2_{\H} \leq 1.$$
\end{proof}

\begin{remark}
The equivalence above comes from the fact that $A^x$, cf. Section \ref{app:SOO}, identified as a linear operator $\H \to \H$, is {\it compact} due to the "good rate function" assumption. In fact, one can see with a little more work that, under assumptions (A3)-(A5), $A^x$ is even Hilbert-Schmidt (in the SDE context this is classical cf. e.g. \cite[Lemme 1.9]{arous1988methods}). This can be proven as a relatively straightforward consequence of properties of the 2nd Wiener chaos (we refrain from giving details since we will not need to use this property).
\end{remark}

\begin{proposition} \label{prop:ND}
Let $A^x = D^2\Phi_1(\h^x) $, then
 $$\I'(x)  A^x < {\rm Id} \text{ (as form on $\TH$)}  \ \Leftrightarrow \ \  \exp\left(\I'\left( x\right)\Delta_{2}\right) \in L^{1+}$$
\end{proposition}

\begin{proof} 
Note that $ \exp\left(\I'\left( x\right)\Delta_{2}\right) \in L^{1+}$ if and only if  
\[
(*): \ \  \mbox{there exists $C>\I^{\prime }\left( x\right)$ such that,  for all $r$ large enough} \ P\left( \Delta _{2}^{x} \ge r\right) \leq \exp (-Cr ).
\]%
Recall that\[
\Delta _{2} = G_2(\VV) - K_2,
\]%
where by assumption $G_2$ is quadratic, so that $\varepsilon ^{2}G_2(\VV)  =  G_2(\delta_{\varepsilon} \VV)$ where $G_2$ is a continuous
function. By Lemma \ref{lem:AboutV} (iii), we know that $ \delta_{\varepsilon} \VV$ satisfies a LDP, so that by the contraction principle one has 
\[
P\left( \Delta_2 \geq \varepsilon^{-2}\right) \sim P\left(\varepsilon^{2} G_2(\VV)  \ge 1\right) = \exp \left( -\frac{%
C^{\ast }+o(1)}{\varepsilon ^{2}}\right) 
\]%
where
\[
C^{\ast }=\inf_{\h \in \H_0}\left\{ \frac{1}{2}\left\Vert
\h\right\Vert ^{2}:\frac 1 2 \partial _{\varepsilon }^{2}|_{\varepsilon=0}\Phi \left(
\varepsilon \h+\h^{x}\right) \ge 1\right\}
\]
Hence $(\ast )$ above is reduced ($r=1/\varepsilon ^{2}$....)\ to the question 
$C^{\ast }>\I^{\prime }\left( x\right) $. 
 However, by Lemma \ref{lem:NDstrong}, if the non-degeneracy assumption holds, then one has for each $\h \in \H_0$
\[ \frac{1}{2}\left\Vert\h\right\Vert ^{2} \geq  \beta^{-1}  \I'(x) \frac 1 2 G_2(\hh) =\beta^{-1} \frac 1 2 \partial _{\varepsilon }^{2}|_{\varepsilon=0}\Phi \left(
\varepsilon \h+\h^{x}\right) \]
so that
\[
C^* \geq \beta^{-1}  \I'(x)  >  \I^{\prime } (x).\]
On the other hand, if the non-degeneracy condition fails there exists $\h$ with $1 =\frac 1 2 \A^{x}\left( \h,\h\right)  \geq \frac{1}{2 \I^{\prime }\left( x\right)} \|\h\|^2$ which implies that $C \leq \I^{\prime }\left( x\right)$.

\end{proof}

\begin{proposition} \label{prop:Jloc1} 
For $x>0$ we have 
$\eps \to 0$
$$
        J (\eps, x)   \sim  \frac{\eps  \beps^2}{\sigma_x \Lambda'(x)^2\sqrt{2\pi}  } E \left[ \exp\left(\I'\left( x\right)\Delta^x_{2}\right) \right] \; ;
$$
whereas for $x<0 $ we have
$\eps \to 0$
$$
        J^{\mathrm{put}} (\eps, x)   \sim  \frac{\eps  \beps^2}{\sigma_x \Lambda'(x)^2\sqrt{2\pi} }  E \left[ \exp\left(\I'\left( x\right)\Delta^x_{2}\right) \right].
$$
\end{proposition}

\begin{proof}
In this proof we denote by $C$ positive constants, whose value may change from line to line.

By Proposition \ref{prop:Jloc} we can work with the sub-probability $P_\delta = P ( ... ; \beps ||| \WW ||| < \delta)$, i.e. on the part of the probability space where the model 
remainder estimates are available. We want to estimate
\begin{eqnarray*}
         J_\delta (\eps, x) & = & E_\delta \left[ \exp \left( {  - \frac{\I'(x)g_1^x }{\beps}  } \right) (e^{ \eps g_1^x + \eps\beps g_2^x + (\eps / \beps) r_3^{x,\eps}} - 1)^+\right]  \\
 \end{eqnarray*}
 
Recall $\eps = \beps / \nu_\eps$. Recall also $r_3^{x,\eps} (\omega) = R_3^{x,\eps} ( \WW (\omega))$. For this 
``robustified'' remainder Assumption (4b) applies so that $ | R_3^{x,\eps} (\WW) | \lesssim o( \beps^2) + \beps^3 ||| \WW |||^3 \ $ 
whenever $\beps ||| \WW ||| \le \delta $.
Thus, for $\beps$ small enough (depending on $\delta$), and for a suitable positive constant $C$,
$$
 | R_3^{x,\eps} (\WW) | \le C \delta \beps^2 + \beps^3 ||| \WW |||^3 \le \delta \beps^2 ( C + ||| \WW |||^2) \ .
$$
Since we also have $\beps g_1(\omega) = \beps G_1 ( \WW (\omega) ) = O(\delta)$ and  $\beps^2 g_2(\omega) = \beps^2 G_2 ( \WW (\omega) ) = O(\delta^2)$ when working with $P_\delta$, 
{on this set we have,  for some constant $C>0$,
\[e^{ \eps g_1^x + \eps\beps g_2^x + (\eps / \beps) r_3^{\eps,x}}
\in 
(1+ \eps g_1^x + \eps\beps g_2^x + (\eps / \beps) r_3^{\eps,x})(1\pm C\delta)
\]
and as a consequence for some constant $C>0$ one has}

\begin{eqnarray*}
         J_\delta(\eps,x)
         & \in & \eps (1\pm C \delta) E_\delta \left[ \exp \left( {  - \frac{\I'(x)g_1 }{\beps}  } \right) ( g_1 + \beps 
         [g_2 \pm \delta (C+ ||| \WW |||^2)])^+\right] .
\end{eqnarray*}
Recall by Lemma \ref{lem:AboutV} one has
$g_2 = \Delta_2 + g_1 \Delta_1 + g_1^2 \Delta_0$
where the $\Delta_i$ are independent from $g_1$, and we let
{\begin{eqnarray*}
\tilde \Delta_{0}  := |\Delta _{0}| + C \delta \| \v \|^2_{\CM}, \\
\tilde \Delta_{0}^\pm  := \Delta _{0} \pm C \delta \| \v \|^2_{\CM}, \\
\tilde \Delta^\pm_{2} := \Delta _{2} \pm \delta (C+ ||| \VV |||^2), 
\end{eqnarray*}
where $\tilde \Delta _{2}^\pm$ is also $P$-independent from $g_1$. Note also that, using $\beps g_1=O(\delta)$,
$$|\beps \Delta_1|= | G_2^0\left(\delta_{\beps} \VV, \v\right)| \leq C \beps ||| \VV ||| \leq C \delta$$
when $\beps |||\WW||\leq \delta$. We also have $|||\WW||| \lesssim |||\VV|||+\|\v\|_{\H} |g_1|$.
Thus, the asymptotic behaviour of $J_\delta (\eps, x)$ is sandwiched by $\eps(1\pm C \delta)$ times
\begin{eqnarray*}
        \dots & \in & E_\delta \left[ \exp \left( {  - \frac{\I'(x)g_1 }{\beps}  } \right) \Big(   g_1+ \beps ( \tilde \Delta^\pm_{2}+g_{1}\Delta_1 + g_{1}^{2}
        \tilde \Delta^\pm_{0})   
        \Big)^+\right]  \\
        & \in & E_\delta \left[ \exp \left( {  - \frac{\I'(x)g_1 }{\beps}  } \right) \Big(   g_1+ \beps  \tilde \Delta^\pm_{2} 
+g_1(\beps \Delta_1+\beps g_1 \tilde{\Delta}_0^\pm)        
        \Big)^+\right] \\
       & \in & E_\delta \left[ \exp \left( {  - \frac{\I'(x)g_1 }{\beps}  } \right) \Big(   g_1+ \beps  \tilde \Delta^\pm_{2} \pm C (1+ \tilde \Delta _{0} ) \delta |g_1| 
        \Big)^+\right] 
\end{eqnarray*}}

We now prove the upper bound for the asymptotics. Clearly,
$$
      E_\delta \left[ \exp \left( {  - \frac{\I'(x)g_1 }{\beps}  } \right) \Big(   g_1+ \beps  \tilde \Delta^+_{2} \pm C (1+ \tilde \Delta _{0} ) \delta |g_1| \Big)^+\right] \le E \big[ \cdots \big] 
$$
where $\cdots$ means the same argument. Set
$$\gamma_\delta:= C(1+\tilde \Delta_0) \delta/ \sigma_x$$
and assume that $\delta$ is small enough that $\gamma_\delta < 1$. By Theorem \ref{th:expansion}, part (iii), we have $\frac{\bar{\eps}}{\Lambda'(x) \sigma_x}>0$ and can then apply Lemma \ref{lem:BSabs} (with $N= g_1 / \sigma_x$) to see that 
$$
E \big[ \cdots  |\Delta_2, \mathbf{V} \big]  \leq 
\frac{\beps^2}{\sigma_x \Lambda'(x)^2\sqrt{2\pi} } \max\left[ (1-\gamma_{\delta}) e^{\frac{\I^{\prime }\left( x\right)\left(\Delta_{2}+ \delta \left(C+||| \VV |||^2\right)\right)}{1-\gamma_{\delta}}} + 2 \gamma_{\delta} , (1+\gamma_{\delta}) e^{\frac{\I^{\prime }\left( x\right)\left(\Delta_{2}+ \delta \left(C+||| \VV |||^2\right)\right)}{1+\gamma_{\delta}}}\right].
$$
By Proposition \ref{prop:ND} and Assumption (A5c), $\exp\left(\I'\left( x\right)\Delta_{2}\right) \in L^{1+}$ and by Lemma \ref{lem:AboutV} (iv) $\exp(||| \VV |||^2) \in L^{0+}$, so that by letting successively $\beps$ and $\delta$ go to $0$ we obtain that 
\begin{equation*}
\limsup_{\varepsilon \to 0} \beps^{-2} E \big[ \cdots \big]  \leq 
{ \frac{1}{\sigma_x \Lambda'(x)^2\sqrt{2\pi} } }
 E \left[ \exp\left(\I'\left( x\right)\Delta_{2}\right) \right].\end{equation*} 
The lower bound is proved in the same way using the lower bound in Lemma \ref{lem:BSabs}. 
The case $x<0$, $J^{\mathrm{put}}(\eps,x)$ can be threated analogously, with the difference that we have to compute
\begin{eqnarray*}
         J_\delta^{\mathrm{put}} (\eps, x) & = & E_\delta \left[ \exp \left( {  - \frac{\I'(x)g_1^x }{\beps}  } \right) (1-e^{ \eps g_1^x + \eps\beps g_2^x + (\eps / \beps) r_3^{\eps,x}})^+\right]  \\
 \end{eqnarray*}
that we can reduce to the computation of
\[
      E_\delta \left[ \exp \left( {  - \frac{\I'(x)g_1 }{\beps}  } \right) \Big(   - g_1 -  \beps  \tilde \Delta^\pm_{2} \pm C (1+ \tilde \Delta _{0} ) \delta |g_1| \Big)^+\right] . 
\]
We apply again Lemma \ref{lem:BSabs}, this time with $\eps = -\frac{\bar{\eps}}{\Lambda'(x) \sigma_x}>$ 
and $N=-g_1/ \sigma_x $. The results follows then as for the call case.

\end{proof}




\subsection{Proof of Theorem \ref{thm:mdp}} \label{proof:MD}

The proof is similar to that of Theorem \ref{thm:main0} (but simpler  since we only need to expand to first order), and we keep the same notations. By Lemma \ref{lem:hx}, for $x$ small enough the minimizer $\h^x$ is unique and is $C^2$ as a function of $x$, in particular $\Lambda$ is $C^1$ at $0^+$. Note also that equation \eqref{eq:defv} and the fact that $g^x_1 = \left\langle D\Phi(\h^x),W\right\rangle$ imply that 
\[ \sigma_x\|  \v ^x \| = 1. \]
%

The proof proceeds as in the large deviation case, and after the same Girsanov transform, we are left with 
\begin{align*}
E \left[ \exp\left(- \frac{\Lambda'(x_\varepsilon) \sigma_{x_\varepsilon} N_1}{\beps}\right) \left(e^{\eps \sigma_{x_\varepsilon} N_1 + \eps \beps R_2^x}-1 \right)^+ \right]
\end{align*}
where $N_1 \sim \mathcal{N}(0,1)$.
Note that on $\beps |||  \mathbf{W} ||| \leq 1$ we have (uniformly in $x$ near $0$)
\begin{equation}
R_2^x \leq C (1+ ||| \mathbf{W} |||^2) \leq  C (1+ ||| \mathbf{V}^x|||^2 + \sigma_x |N_1| \|\v^x\|_{\CM} )
\end{equation}
and in addition by the LDP for $\mathbf{W}$ the term 
$$E\left[ \ldots 1_{\{ \beps |||  \mathbf{W} ||| \geq 1\}} \right] \leq \exp\left( -\frac{c}{\bar{\varepsilon}^{2}}\right)$$
is negligible compared to what we want. We are therefore left with
\begin{align*}
E \left[ \exp\left(- \frac{\Lambda'(x_\varepsilon) \sigma_x N_1}{\bar{\varepsilon}}\right) \left(\exp \left(\eps \sigma_x N_1 \pm \eps \bar{\varepsilon} C \left(1+\| \mathbf{V}^x\|^2 + \sigma_x |N_1| \|\v^x\| \right)\right) -1\right)^+ \right]
\end{align*}
which by using Lemma \ref{lem:BSabs} as in the LDP case  is equivalent to
\begin{align*}
\eps \frac{\bar{\varepsilon}^2}{\sqrt{2\pi} \Lambda'(x_\varepsilon)} E \left[ \exp \left(\pm C \Lambda'(x_\varepsilon)  \left(1+\|\mathbf{V}^x\right\|^2) \right) \right].
\end{align*}

Finally we use that 
\[\|\mathbf{V}^x\| \lesssim \|\mathbf{W}\| + \|\v^x\| |g_1^x|\leq  \|\mathbf{W}\| +  |N_1|\] 
has Gaussian tails, uniformly in $x$, so that since $\Lambda'(0) = 0$ {(differentiability at $0$ has already been proved)} we obtain
\[\lim_{\eps \to 0} E \left[ \exp \left(\pm C \Lambda'(x_\varepsilon)  \left(1+\|\mathbf{V}^x\|^2\right) \right) \right] =1,
\]
which concludes the proof.

\appendix

\section{Robust representation} \label{sec:Appendix3A}

Assumption (A3a) expresses our ability to write the solution process of interest, with small noise parameter $\beps$, as continuous image of its ($\beps$-rescaled) noise, lifted to a random rough path (model). In the theory of rough paths (regularity structures) this is achieved by rough integration (reconstruction), plus a fixed point argument when it comes to differential equation. These solution theories by nature provide detailed information of the solution and perturbations thereof, many of which play a role in our analysis. We take here a pragmatic view, and simply list the used properties. These are well-known in case of classical SDEs, the robust representation is precisely Lyons' rough path view, a review of selected topics was given in Section \ref{sec:SDEs}. All relevant results for RoughVol are collected in Appendix \ref{app:ERS}. 

\medskip

Our basic assumption (A3a) postulates the existence of a regularity structure, with model space $\M$, such that $\bar X^\eps (\omega)$, defined on $m$-dimensional Wiener space, and in general only measurable, can be written as  
$$ \bar X^\eps (\omega) = 
\Phi^\eps ( \delta_{\beps} \WW (\omega)) \text{ a.s. } $$ 
with continuous map $\Phi^\bullet: (\eps, \MM) \mapsto \Phi^\eps(\MM)$, from $[0,1]\times  \M \to \mathcal{C}([0,1])$, where $\WW (\omega)$ lifts the noise (i.e. $m$-dimensional Brownian motion $\W(\omega) \equiv \omega$) to a random model (typically by attaching additional stochastic integrals, as seen in (\ref{eq:DefinitionWW}) for instance). 
\medskip
%
%
%

\textbf{Model space}: $\mathcal{M}$ is an algebraic subspace of an ambient Banach space given as direct sum (finitely many, indexed by $i=1,2,...$) H\"older type spaces, the first of which is given by $\C^{\alpha}$ and accommodates typical noise realizations. A {\it model}  $\MM$ (over $M \in \C^{\alpha}$) is then simply a tuple of the form $(M, ... ) \in \mathcal{M}$, 
where the dots depend on the problem at hand. Its analytic properties are captured by having $ \| \MM \| = \sum \| \pi_i (\MM) \|_i = \| M \|_\alpha + ...$ finite. This makes $\M$ a complete metric space $\M$ with distance  
$$(\MM,\MM')\mapsto \| \MM;\MM' \| := \| \MM - \MM' \| 
$$
between models.\footnote{In \cite{hairer2014theory} triple bars are used for this model distance. We reserve triple bars for the homogenous (model) norm, in agreement with rough path notation \cite{friz2014course}.} 


\textbf{Canonical lift}: for $M \in \C^{\infty}$ we assume existence of an analytically well-defined lift $\mathcal{L}(M) = (M,...)\in\mathcal{M}$, with dots typically given by classical integration. We assume this extends to elements of Cameron--Martin space $\CM = \H^1 \subset C([0,1],\R^m)$, in the sense that  $\CM \ni \h   \mapsto \mathcal{L}(h) = \hh \in \M$ with the property that for all $R>0$, $ \left\{ \hh: \;\; \|h\|_{\H} \leq R \right\} \mbox{ is compact in } \M$. \\
\textbf{Dilation and homogenous norm}: for each $\eps>0$, there exists an {\it dilation} operator $\delta_\eps : \M \to \M$, which lifts scalar multiplication in the sense,
$$ \forall \h \in \H, \;\;\delta_{\eps} \hh = \mathcal{L}(\eps h) = (\eps^{d_1} \pi_1 (\hh), \eps^{d_2} \pi_2 (\hh), ... ),$$
and integers $d_i$ and $d_1=1$. A {\it homogenous model norm} is given by  $||| \MM ||| = \sum \| \pi_i (\MM) \|_i^{1/d_i}$ so that 
$$ \forall \MM \in \M, \;\;||| \delta_\eps \MM||| = \eps |||\MM|||,$$

\textbf{Renormalization} : we are given a group $\mathfrak{G}$ which acts continuously on $\M$, and which preserves the first level in the sense that $\pi_1 = \pi_1 \mathfrak{R}$ for each $\mathfrak{R}$ in $\mathfrak{G}$, commutes with translation (see below), and provides an Ito-Stratonovich typ correction 
  for approximating Gaussian models (see below). 

\textbf{Translation operator} : we have a continuous map $\H \times \M \to \M$, $(h,\MM) \mapsto T_h \MM$ such that 
$T_h\mathcal{L}(k) = \mathcal{L}(h+k)$ for all $h,k \in \H$. Furthermore, 
\begin{equation} \label{eq:estThM}
\exists C>0, \;\; \forall (h,\MM) \in \H \times \MM, \;\;|||T_h \MM ||| \leq C |||\MM|||+ C \|\h\|_{\H},
\end{equation}
and 
\begin{equation}
\forall (h,\mathfrak{R}) \in \H \times \mathfrak{G}, \;\;\;\; T_h \mathfrak{R} = \mathfrak{R}T_h .
\end{equation}

\textbf{Gaussian model} : we finally have some probabilistic assumption. Given $\eta>0$, let $\W^\eta(\omega) = \W(\omega) \ast \rho^\eta$, where $\rho^\eta$ is a smooth approximation of identitity (i.e. $\rho^\eta=\eta^{-1}\rho(\eta^{-1}\cdot)$ where $\rho$ is smooth, compactly supported and with $\int \rho = 1$). Note that $\W^\eta$ is smooth, so that one can define its canonical lift $\WW^\eta$. We then assume that there are some (deterministic) elements $\mathfrak{R}^\eta$ of $\mathfrak{G}$ and a (random) model $\WW(\omega) \in \M$ s.t. 
$$\lim_{\eta \to 0} \mathfrak{R}^\eta \WW^\eta = \WW$$
in probability and model topology; as well as consistency with Cameron-Martin shifts, namely that
$$\forall \h \in \H, \WW(\omega+\h) = T_h \WW(\omega), \;\; \mbox{ a.s.}$$ 

Note that combined with (\ref{eq:estThM}), this implies, via a generalized Fernique estimate \cite{friz2010generalized}, that $|||\WW|||$ has Gaussian tails, namely
\begin{equation} \label{eq:Fernique}
\exists \gamma >0, E \exp\left( \gamma |||\WW(\omega)|||^2 \right) < \infty.
\end{equation}

\section{Elements of Regularity Structures for Rough Volatility} \label{app:ERS}


\subsection{The rough vol model}  \label{subsec:rsrv}

We review the essentials of \cite{bayer2019regularity}, at first in the setting of Hairer \cite{hairer2014theory}.


\subsubsection{Basic pricing setup}
\label{sec:basic-pricing-setup}
Recall $\hat W_t =\int^t_0 K^H (s,t)  dW_s$, with $K^H (s,t) = \sqrt{2H}  |t-s|^{H-1/2},\ H \in (0, 1/2]$. Given a (sufficiently) smooth scalar function $f$, we are interested in robust integration of $\int f(\widehat{W}) dW$. The H\"older exponent of $\hat W$ is $H-\kappa$, any $\kappa>0$. Let $M$ be the smallest integer such that
$ (M+1)H-1/2>0$ and then pick $\kappa$ small enough such that
\begin{align}  \label{eq:condonkappa}
(M+1)(H-\kappa)-1/2-\kappa >0\,.
\end{align}
(In case $H=1/2$, have $M=1$ and then $1/2-\kappa \in (1/3,1/2)$ which is the rough path case.) 
The  \textit{structure space}\footnote{We here follow the terminology of the forthcoming 2nd edition of \cite{friz2014course}.} is defined as 
\begin{align}     \label{equ:SimpleModelSpace}
\mathcal{T}=\left\langle \{ \Xi, \Xi \mathcal{I}(\Xi),\ldots,\Xi \mathcal{I}%
(\Xi)^M,\mathbf{1},\mathcal{I}(\Xi),\ldots,\mathcal{I}(\Xi)^M \} \right\rangle
\,,
\end{align}
where $\langle\ldots \rangle$ denotes the vector space generated by the
(purely abstract) symbols in $\{\ldots \}=: S$. 
%
%
%
The
symbol $\mathcal{I}(\ldots)$ represents ``integration against the
 kernel $K^H$'', so that $\mathcal{I}(\Xi)$ represents fractional Brownian motion $\hat{W}$. Symbols
like $\Xi \mathcal{I}(\Xi)^m =\Xi \cdot \mathcal{I}(\Xi)\cdot \ldots \cdot 
\mathcal{I}(\Xi)$ or $\mathcal{I}(\Xi)^m=\mathcal{I}(\Xi)\cdot \ldots \cdot 
\mathcal{I}(\Xi)$ should be read as products between the objects above.
Every symbol $\tau \in S$ has a {\it homogeneity}, given by 
\begin{align*}
|\Xi \mathcal{I}(\Xi)^m| &= - 1/2 - \kappa + m (H-\kappa),\,m\geq 0 \\
|\mathcal{I}(\Xi)^m|&=m (H-\kappa) ,\,m>0 \\
|\mathbf{1}|&=0\,.
\end{align*}
Introduce the set of homogeneities $%
A:=\{|\tau|\,\vert\, \tau \in S\}$, with minimum $|\Xi|=-1/2-\kappa$.
Note that the homogeneities are  
multiplicative in the sense that, $|\tau\cdot \tau^{\prime }| = |\tau| + |\tau^{\prime }|$
for $\tau,\tau^{\prime }\in S$. 


At last, we have the {\it structure group} $G$, an (abstract) group of linear operators on the model space $\mathcal{T}$ which should satisfy $\Gamma \tau- \tau \in \bigoplus_{\tau'\in S:\,|\tau'|<|\tau|} \RR\tau'$ and $\Gamma \mathbf{1}=\mathbf{1}$ for $\tau \in S$ and $\Gamma \in G$. We will choose $G=\{\Gamma_h \,\vert\, h\in (\RR,+)\}$ given by 
\begin{align*}
\Gamma_h \mathbf{1}=\mathbf{1},\,\Gamma_h \Xi=\Xi,\,\Gamma_h \mathcal{I}(\Xi)=\mathcal{I}(\Xi)+h \mathbf{1}\,. 
\end{align*}
and $\Gamma_h (\tau'\cdot \tau)=\Gamma_h \tau' \cdot \Gamma_h \tau$ for $\tau',\,\tau \in S$ for which $\tau\cdot \tau'\in S$ is defined. 
\bigskip

\subsubsection*{The It\^o model $(\Pi,\Gamma)$}


To give a meaning to the product terms $\Xi\mathcal{I}(\Xi)^k$ we follow the
idea from rough paths and define ``iterated integrals'' for $s,t\in 
\mathbb{R} , s\leq t$ as 
\begin{align}  \label{eq:DefinitionWW}
\mathbb{W}^m(s,t)=\int_s^t (\hat{W}(r)-\hat{W}(s))^m\,\mathrm{d}  W(r) \ .
\end{align}
%
%
%
%
%

We are now in the position to define a model $(\Pi,\Gamma)$ that gives a
rigorous meaning to the interpretation we gave above for $\Xi,\mathcal{I}%
(\Xi),\Xi\mathcal{I}(\Xi),\ldots\,$. Recall that in the theory of regularity
structures a model is a collection of linear maps $\Pi_s: \mathcal{T}%
\rightarrow C^1_c(\mathbb{R} )^{\prime }$, $\Gamma_{st}\in G$ such that 
\begin{align}
&\Pi_t =\Pi_s \Gamma_{st},\,  \label{eq:ModelDefinition1} \\
&|\Pi_s \tau (\varphi^\lambda_s)| \lesssim \lambda^{|\tau|}\,,
\label{eq:ModelDefinition2} \\
&\Gamma_{st} \tau =\tau+ \sum_{\tau^{\prime }\in S:\,|\tau^{\prime }|<|\tau|}
c_{\tau^{\prime }}(s,t) \tau^{\prime },\,|c_{\tau^{\prime }}(s,t)|\lesssim
|s-t|^{|\tau|-|\tau^{\prime }|}  \label{eq:ModelDefinition3}
\end{align}
where the bounds hold uniformly for $\tau\in S$, any $s,t$ in a compact set
and for $\varphi^\lambda_s:=\lambda^{-1}\varphi(\lambda^{-1} (\cdot-s))$
with $\lambda\in(0,1]$ and $\varphi\in C^1$ with compact support in the ball 
$B(0,1)$. We define the following ``It\^o'' model $(\Pi,\Gamma) = (\Pi^{\text{It\^o}},\Gamma^{\text{It\^o}})$. 
\label{LimitingModel}
\begin{equation*}
\begin{array}{ll}
\Pi_s \mathbf{1}=1 & \Gamma_{ts} \mathbf{1} =\mathbf{1} \\ 
\Pi_s \Xi = \dot{W} \qquad & \Gamma_{ts} \Xi =\Xi \\ 
\Pi_s \mathcal{I}(\Xi)^m = \left(\hat{W}(\cdot)-\hat{W}(s)\right)^m & 
\Gamma_{ts} \mathcal{I}(\Xi)=\mathcal{I}(\Xi)+(\hat{W}(t)-\hat{W}(s))\mathbf{%
1} \\ 
\Pi_s \Xi \mathcal{I}(\Xi)^m = \{ t \mapsto \frac{\mathrm{d} }{\mathrm{d}  t}\mathbb{W}%
^m(s, t) \} & \Gamma_{ts} (\tau \tau^{\prime })=\Gamma_{ts}\tau \cdot
\Gamma_{ts}\tau^{\prime }\,,\,\, \mbox{ for } \tau,\tau^{\prime }\in S\mbox{
with } \tau\tau^{\prime }\in S%
\end{array}
\end{equation*} 
We extend both maps from $S$ to $\mathcal{T}$ by imposing linearity.

\begin{lemma} \cite{bayer2019regularity}
\label{lem:LimitingModel} The pair ($\Pi^{\text{It\^o}},\Gamma^{\text{It\^o}})$ defines
(a.s.) a model on $(\mathcal{T},A)$.
\end{lemma}

\subsubsection{Full regularity structure for rough volatility}
\label{sec:full-pricing-setup}
Rough volatility is specified in terms of two independent Brownians $(W,\bar W)$. Writing $\bar{\Xi}$ for the abstract symbol that corresponds to (distributional) derivative of $\bar W$ this leads to 
%
\begin{align}
\bar{\mathcal{T}} = 
 \mathcal{T} + \left\langle \{ \bar{\Xi}, \bar{\Xi} \mathcal{%
I}({\Xi}),\ldots,\bar{\Xi} \mathcal{I}(\Xi)^M \}\right\rangle .
\end{align}
Again we fix $|\bar{\Xi}| =-1/2-\kappa$ and the homogeneity of the other
symbols are defined multiplicatively as before. We extend the It\^o model $(\Pi,\Gamma)$ to this regularity structure by
defining 
\begin{equation*}
\Pi_s \bar{\Xi} \mathcal{I}(\Xi)^m  =
\left\{ t \mapsto 
 \frac{\mathrm{d} }{\mathrm{d}  t}
\left(\int_s^t \left(\hat{W}(u) - \hat{W}(s)\right)^m d\bar{W}(u) \right) \right\}
\end{equation*}
(the above integral being in It\^o sense), and \footnote{Upon setting $\Gamma_{ts}\left( \bar{\Xi} \right) = \bar{\Xi}$, the given relation is precisely implied by multiplicativity of $\Gamma$.}
\begin{equation*}
\Gamma_{ts}\left( \bar{\Xi} \mathcal{I}(\Xi)^m\right) = \bar{\Xi}
\Gamma_{ts}\left( \mathcal{I}(\Xi)^m\right).
\end{equation*}
Lemma \ref{lem:LimitingModel} extends and yields a random model (``It\^o model'') on $\bar{\mathcal{T}}$.
Hairer's reconstruction theorem (\cite{hairer2014theory}, \cite[Sec. 13]{friz2014course}) then yields a robust view on It\^o-integration and, in particular, identifies objects such as $\int \sigma ( \hat W) d \tilde W$ and asset price give as stochastic exponential of $\int \sigma ( \hat W) d \tilde W$, with $ \tilde W = \rho W + \bar \rho \bar W$ as continuous function $\Phi$ of the It\^o model.

%

\subsubsection{Rough path like formalism} \label{ss}

We note that the It\^o model $(\Pi^{\text{It\^o}},\Gamma^{\text{It\^o}})$ constructed above is in one-one correspondence with objects of the form
$$ \WW (\omega) =  \left(W, \bar W,  \int \hat W dW , \int \hat W d \bar W,  \int \hat W^2 dW, .... , \int \hat W^M d \bar W\right)  \ .
$$
The {\it model norm} on $[0,T]$, usually taken as smallest (uniform) constant implicit in the estimates  (\ref{eq:ModelDefinition2}),(\ref{eq:ModelDefinition3}),
can now simply be written as in terms of generalized H\"older norms,
\begin{equation} \label{def:RPnorm}
        \| \WW \| :=    \| W \|_{1/2-\kappa} + \| \bar W \|_{1/2-\kappa} + \Big\|  \int \hat W d \bar W \Big\|_{H +1/2 - 2\kappa}  ... + \Big\|  \int \hat W^M d \bar W \Big\|_{M(H-\kappa)+1/2 - \kappa} \ ,
\end{equation} 
where, e.g. $ \| W \|_{1/2-\kappa} = \sup_{0\le s < t \le T} |W_{s,t}|/|t-s|^{1/2-\kappa}$ and the final summand 
 is spelled out as
$$
              \sup_{0\le s < t \le T} \frac{\Big|  \int_s^t (\hat{W}_r-\hat{W}_s)^M d \bar W_r \Big|}{|t-s|^{M(H-\kappa)+1/2 - \kappa}} \ .
$$
The {\it model metric} is given by $\| \WW ; \VV \| := \| \WW - \VV \|$.  We insist that, from an rough analysis perspective, $\int \hat W^m dW$, denoted by $\mathbb{W}^{m}$ in \cite{bayer2019regularity}, is a priori supplied as analytic object, even though what we have in mind is a typical realisation of the corresponding stochastic objects: as shown in \cite{bayer2019regularity} and reviewed above (Lemma \ref{lem:LimitingModel}),  a.s. $\WW (\omega)$ has the correct regularity, i.e. $ \| \WW  (\omega)\| < \infty$ a.s.

The reconstruction theorem \cite{hairer2014theory,bayer2019regularity} shows that
It\^o integration (and much more ...) is locally Lipschitz in this metric. More precisely, for sufficiently smooth $f$,
$$
            \Big(\int f ( \hat W) d \W \Big) (\omega) = \int f ( \hat W (\omega)) d \WW(\omega) \text{ \ a.s. } 
$$
where the left-hand side is a classical It\^o integral, and the right-hand side defined (thanks to reconstruction) in a robust pathwise fashion, as locally Lipschitz function of $\WW = \WW (\omega)$  .

\subsection{Homogenous model norms and model translation} \label{app:HMN}  \label{app:mt}


From a Gaussian concentration perspective, it is preferable to work with the  {\it homogenous model norm} given by 
\begin{equation}\label{def:hom:norm}
        ||| \WW |||  :=    \| W \|_{1/2-\kappa} + \| \bar W \|_{1/2-\kappa} + \Big\|  \int \hat W d \bar W \Big\|^{1/2}_{H +1/2 - 2\kappa} ... + \Big\|  \int \hat W^M d \bar W \Big\|^{1/(M+1)}_{M(H-\kappa)+1/2 - \kappa} \ ,
\end{equation}


We also need a generalization of the translation map known from rough paths (see e.g. \cite{friz2014course} and \cite{cannizzaro2017malliavin} for the case of model translation in the context of singular SPDEs) to the rough volatility regularity structure \cite{bayer2019regularity}. Indeed, one has to be careful in ``lifting'' the meaning of ``$\W  + \h$''  to a rough path or model. Formally,
%
with $\h = (h,\bh)$,
$$
         T_{\h}  (\WW) =  \left(W+h, \bar W+ \bh, \hat W+ \hat h, \int (\hat W +  \hat h) d(W+\bh), ... \right) \ .
$$
Note that $T_{-\h} (\WW) \ne   %
         \WW - \mathbf{h} =   (W- h, \bar W- \bh, \hat W- \hat h, \int \hat W  dW  - \int \hat h d\bh, ...)$,
which appears in the model distance between $\WW$ and $\mathcal{L}(\h)= \mathbf{h}$, for $\h \in \CM$.

\medskip

The following lemma can be seen as variation of rough path results found in \cite[Sec 11.1]{friz2014course}. In particular, we take $T_\h$ as in the formal discussion above, and check below that all cross integrals involving $\h$ are (analytically) well-defined. Write $\M$ for the model space, the (complete metric space) of all $\WW$ with $\| \WW \| < \infty$, subject to algebraic constraints (generalized Chen relations), imposed by the additivity of all integrals used in defining $\WW$ over smooth $\W$, see \cite[Lem. 3.2.]{bayer2019regularity}.

%
%
\begin{lemma} \label{lem:RVMtranslation}
 (i) We have a continuous map $\H \times \M \to \M$, $(h,\MM) \mapsto T_h \MM$ such that 
$T_h\mathcal{L}(k) = \mathcal{L}(h+k)$ for all $h,k \in \H$ and $T_h \circ T_{-h}$ is the identity map on $\M$. \\
(ii) There exists $C>0$ s.t. for all $ (h,\MM) \in \H \times \M$,
\begin{equation} 
|||T_h \MM ||| \leq C |||\MM|||+ C \|\h\|_{\H} \;.
\end{equation}
\end{lemma} 
\begin{proof} 
(i) Recall $\h \in \CM \subset C^{1/2}$, so that $W- h, \bar W- \bh$ are again $(1/2-\kappa)$-H\"older. More interestingly,
$$
      \hat W - \hat h = \{ t \mapsto \hat W_t - \int_0^t K^H(t,s) \dot h_s ds \}
$$
is again $C^{H-\kappa}$ using that $K^H$ convolution not only maps $C^{-1/2-\kappa} \to C^{H-\kappa}$ (which explains $\hat W \in C^{H-\kappa}$), but also $L^2 = W^{0,2 } \to  W^{1/2+H,2} \subset C^H$, so that $\hat h \in C^H$.

We now move to higher level objects, such as $  \int \hat W^k  dW $, with translate given by 
$$
       \int (\hat W + \hat h)^k d(W + h) =  \sum_{j=0}^k {k\choose j} \left( \int \hat W^j \hat{h}^{k-j}  dW + \int \hat W^j \hat{h}^{k-j}  dh \right).
$$
We need to check a certain H\"older type regularity on the arising integrals. Let us explain for example how to prove
$$\int_s^t \hat h_{s,r}^k \hat{W}_{s,r}^l dW_r \lesssim \left|t-s\right|^{(k+l)H + 1/2 -(k+l+1) \kappa}$$ 
for positive integers $k$ and $l$ and small fixed $\kappa$.
%

We first prove that 
\begin{equation} \label{eq:pvar}
\left\|\int_s^\cdot \hat{W}_{s,r}^l dW_r \right\|_{p-var,[s,t]} \lesssim |t-s|^{l H + \frac 1 2 - (l+1) \kappa}
\end{equation}
for $p \geq (1/2-\kappa)^{-1}$.

Indeed, for a given partition $\{t_i\}$ of $[s,t]$, one has
\begin{align*}
\left|\int_{t_i}^{t_{i+1}} \hat{W}_{s,r}^l dW_r \right|^p &\lesssim \left|\int_{t_i}^{t_{i+1}} \hat{W}_{t_i,r}^l dW_r \right|^p + \left(\sup_{r \in [s,t]} \left|\hat{W}_{s,r}^l- \hat{W}_{t_i,r}^l\right|\right)^p \left| \int_{t_i}^{t_{i+1}} dW_r\right|^p \\
&\lesssim \left|t_{i+1}-t_i\right|^{p(lH+\frac{1}{2}-(l+1)\kappa)} + \left|t-s\right|^{p(lH-l\kappa)}  \left|t_{i+1}-t_i\right|^{p(\frac{1}{2}-\kappa)} 
\end{align*}
and summing over $i$ we obtain \eqref{eq:pvar}. We then use the fact that, by the Besov variation embedding of \cite{friz2006variation}, there exists $q$ with $q+p>1$ such that 
$$ \| \hat h  \|_{q-{\rm var};[s,t]} \lesssim |t-s|^{H-\kappa}  \| h  \|_{H^1([0,T])}$$
so that by Young's inequality it holds that 
$$
              \left| \int_s^t \hat h_{s,r}^k \hat{W}_{s,r}^l dW_r \right| \lesssim \| \hat h^k_{s,\cdot}  \|_{q-{\rm var};[s,t]} \left\| \int_s^\cdot \hat{W}_{s,r}^l dW_r \right\|_{p-{\rm var},[s,t]} \lesssim \left|t-s\right|^{(k+l)H +\frac{1}{2} - (k+1)\kappa}. $$

Continuity on model distance of the translation operator is an easy consequence of the above analysis. 

(ii) is straightforward by keeping track of which powers of $\|h\|$ and $||| \WW |||$ appear in the computation, for instance in the case studied above one has 
$$\left(\frac{\left| \int_s^t \hat h_{s,r}^k \hat{W}_{s,r}^l dW_r \right|}{\left|t-s\right|^{(k+l)H +\frac{1}{2} - (k+1)\kappa}}\right)^{\frac{1}{k+l+1}} \lesssim \| \h \|_{\CM}^{\frac{k}{k+l+1}} |||\WW|||^{\frac{l}{k+l+1}} \lesssim  \| \h \|_{\CM}+ |||\WW|||.$$

%
\end{proof}

%

\begin{lemma}[Robust Cameron--Martin shift]
Let $\WW (\omega)$ be the It\^o model built in Section \ref{sec:full-pricing-setup} above two-dimensional Brownian motion. Then, for all $\h \in \H$, with probability one, $\WW (\omega + \h) = T_{\h} \WW (\omega)$.
\end{lemma} 

\begin{proof} 
This is similar to Theorem \cite[Thm 11.5]{friz2014course} in the rough path setting or \cite{cannizzaro2017malliavin} in the setting of the gPAM model.
\end{proof}

%
%

\begin{proposition} \label{prop:B4}Write as before $\mathcal{L}(\h) = \mathbf{h}$. For all $R>0$, 
$$ \left\{ \mathbf{h}: \;\; \|h\|_{\H} \leq R \right\} \mbox{ is compact in } \M. $$ 
\end{proposition}
\begin{proof} This is clear from the interpretation as good rate function for $\WW (\omega)$, 
cf. \cite{bayer2019regularity}.
\end{proof}

\subsection{Renormalized Wong-Zakai type result} \label{app:Renorm}

Let $W^\eta$ be the mollification of $W$ obtained by convolutions with a mollifier function at scale $\eta >0$ and similar for $\bar W$. There is a canonical lift to a model $\WW^\eta$, which does {\it not} converge when $H <1/2$
(cf. \cite{bayer2019regularity}) and converges to the Stratonovich Brownian rough path when $H=1/2$. Increments $\WW^\eta_{s,t}$ take values in certain (truncated) Hopf algebra, which - as vector space - we can identify with $\R^{M(H)}$ for some integer $M(H)$. 

\begin{theorem} There exists a family of linear maps on $M^\eta: \R^{M(H)} \to \R^{M(H)}$ so that $\hat \WW$, with pointwise definition
$$
             \hat \WW^\eta_{s,t} := M^\eta \WW^\eta_{s,t}
$$
converges (in probability and model topogy) to the It\^o model $\WW$. Moreover the action of $M^\eta$ commutes with the translation operator from Section \ref{app:mt}.
\end{theorem} 
\begin{proof} 
This is a (non-quantitative) formulation of \cite[Thm 3.14]{bayer2019regularity}. The commutation relation is easy to check by hand and fully consistent with \cite{bruned2017rough}, which identifies (in a general branched rough path context) renormalization with higher order translation, with resulting {\it abelian} renormalization group.
\end{proof} 

%

\subsection{``Stochastic'' Taylor remainder estimates via model norms} \label{sec:RoughRemainder} 

We need a variation of the rough path results \cite{aida2007semi, inahama2007asymptotic, inahama2008laplace, inahama2010stochastic} in the 
setting of regularity structures for rough volatility, as recalled in section \ref{subsec:rsrv}. Recall from \cite{bayer2019regularity} that there is a well-defined dilation $\delta_\eps$ acting on the relevant models. Formally, it is obtained by replacing each occurance of $W, \bar W, \hat W$ with $\eps$ times that quantity. As a consequence, dilation works well with homogenous model norms,
$$
 ||| \delta_\eps \WW ||| = \eps ||| \WW ||| \ .
$$

The following theorem is purely deterministic. 

\begin{theorem}[Stochastic Taylor-like expansion] \label{thm:STLE}
Let $f$ be a (sufficiently
) smooth function. Fix $\h \in H^1$ and $ \eps >0$. If $\WW$ is a model (as in the previous section), then so is $T_{\h} ( \delta_{\eps} \WW)$. The pathwise ``rough/model'' integral 
$$
\Psi(\eps) :=  \int_{0}^1 f \left( \eps \hat W_t + \hat{h}_t \right) 
d (T_{\h} ( \delta_{\eps} \WW))_t \ 
$$
is well-defined, continuously differentiable in $\eps$ and we have the remainder estimate
$$  | \Psi(\eps) - \Psi(0) - \eps \Psi'(0) - (1/2) \eps^2 \Psi''(0) | = O (\eps^3 ||| \WW |||^3) \ , $$ valid on bounded sets of $\eps ||| \WW |||$. 
\end{theorem}
\begin{proof}  Because of the similarities with e.g. \cite{inahama2010stochastic} we shall be brief. Formal (but justifiable) differentiation with respect to $\eps$ yields
$$\Psi'(\eps) = \int_0^1 f(\eps \hat{W}_t + \hat{h}_t) d \WW_t  +  f'(\eps \hat{W}_t + \hat{h}_t) \hat{W}_t d( T_\h \delta_\eps \WW_t)$$
and 
$$\Psi''(\eps) = \int_0^1 2 f'(\eps \hat{W}_t + \hat{h}_t) \hat{W}_t d \WW_t  +  f''(\eps \hat{W}_t + \hat{h}_t) \hat{W}_t^2 d( T_\h \delta_\eps \WW_t)$$
and similarly for $\Psi'''$ which we need not spell out.  All integrals here are defined by Hairer's reconstruction (rough integration in case $H=1/2$) i.e. as limit of suitable Riemann-sum approximations involving the "elementary" objects in the model e.g. $\int \hat{W}^k dW$; one also needs to use the regularity of $\h$ as in the proof of Lemma \ref{lem:RVMtranslation}. One then just needs to check that $\Psi'''(\eps) = O (||| \WW |||^3)$, uniformly over $\eps \in [0,1]$, and here one uses precisely the assumption that $ \eps ||| \WW ||| $ remains bounded. 

\end{proof}


\section{Local analysis around the minimizer}

Recall that we are interested in $\Phi: (\eps, \MM) \mapsto \Phi^\eps(\MM)$, from $(0,1]\times  \M \to C[0,1]$. Here we restrict to $\eps = 0$ and $\CM \subset \MM$, where elements $\h \in \CM$ are identified with their canonical lift $\mathcal{L}(\h) = \hh \in \M$. Recall that compactness of $\CM$-level sets in $\MM$ is a standing assumption (Appendix A, see also Proposition \ref{prop:B4}). This entails
a useful regularity property of $\Phi^0$ restricted to Cameron--Martin elements.
\begin{lemma} \label{lem:WeakCont} In the setting of Appendix \ref{sec:Appendix3A},
 let $\Phi^0 $ be a continuous map on $\mathcal{M}$. Then
 $$
      \Phi (\h) :=  \Phi^0 (\hh),
 $$
 seen as map defined on $\H$, in mildly abusive notation, is weakly continuous on bounded sets.
\end{lemma}
\begin{proof}
Let $\h^n$ be a bounded sequence with weak limit $\h \in \H$. We show $\mathcal{L}(\h^n) = \hh^n \to \mathcal{L}(\h) = \hh$ in $\M$. Indeed, by compactness of $\{ \|\h\|_H \leq R \}$ in $\M$, we have along a subsequence $\hh^n \to \kk = \mathcal{L}(\k) $ in $\M$, for some $\k \in \H$. By the structure of $\M$ put forward in Appendix \ref{sec:Appendix3A}, with first level $\C^{\alpha}$ , we have (strong) convergence of $\h^n$ to $\k$ in $\C^{\alpha}$  ($\supset H$). 
 This implies that $\k = \h$, so that any subsequence of $\hh^n$ admits a further subsequence converging to $\hh$, and the claim follows. Weak continuity of $\Phi$ on $\H$ is then immediate.
\end{proof}

From here on we consider $\Phi_1: \H \to \R$, and more specifically $\mathcal{K}^{x} \subset \CM$, the space of $x$-{\it admissible controls}, i.e. elements $\h \in H^1: \Phi_1 (\h) = x$.
 Whenever $\mathcal{K}^{x}$ is non-empty the energy
\begin{equation*}
\I\left(x\right) = \inf_{\h \in \CM}\left\{ \tfrac{1}{2}\int_{0}^{1}|\dot{\h}|^2 dt :\Phi_1  (\h)  =x\right\} =  \inf_{{\rm h} \in \mathcal{K}^{x}} \tfrac{1}{2} \| \h \|^2_{\CM}
\end{equation*}
is finite.

%

\subsection{First order optimality} \label{app:FOO}

In this section we make the standing assumption that $\h \mapsto  \Phi _{1} (\h)$ is (Fr\'echet) $C^1$, and consider a minimizing $x$-admissible control $\h^x$ such that
$$
                      D \Phi _{1} ({\h^x}) \in L ( \CM \to \R ) \ 
$$
is surjective. This entails (cf. \cite[p.25]{bismut1984large}) that $\mathcal{K}^{x}$ to be a Hilbert manifold near $\h%
^{x}$ with tangent space
\begin{equation} \label{equ:tan}
\mathfrak{Ker}D\Phi _{1}\left( \h^x\right) =T_{
\h^{x}}\mathcal{K}^{x}=\left\{ \h\in \CM:
\left\langle D\Phi _{1}\left( \h%
^{x}\right) ,\h\right\rangle
=0\right\} =: {\TH} \ .
\end{equation}%

\begin{lemma}[First order optimality, Lagrange multiplier] \label{lem:appendix-2pre} For each such optimal control $\h^x$ there exists a unique $q^x=q(\h^x) \in \R$ (think: tangent space at $x \in \R$) such that 
\begin{equation*}
\h^x = D \Phi_1 (\h^x)^*q^x
\end{equation*}
where we recall that $D \Phi_1 (\h^x) : \H \to \R$ so that its adjoint maps $\R \to \CM$ where we identify $\R^*, \CM^*$ with $\R,\CM$ respectively. 
 \end{lemma} 

\begin{proof} 
The map $D\Phi _{1}\left( \h^x\right)^*: \R \to \TH^\perp$ is one-one. On the other hand, because $\h^x$ is a minimizer, the differential of $\h \mapsto \frac{1}{2}\|\h\|_{\H}^2$ at $\h^x$ must be zero on $ \TH$, i.e.
$$
           \left\langle \h^x  ,\k \right\rangle = 0 \text{ for all $\k \in \TH$} \ 
$$
so that $\h^x \in \TH^\perp$. We conclude that there exists a (unique) value $q^x \in \R$ s.t. $D\Phi _{1}\left( \h^x\right)^*q^x = \h^x$.
\end{proof} 

Whenever the energy $\I$ is $C^1$ near $x$, we can see that $q^x = \I'(x)$. 

\begin{lemma}[First order optimality and energy] \label{lem:appendix-2} Assume $\I$ is $C^1$ near $x$. \\
 (i) any optimal control $\h^x \in \mathcal{K}^{x}$ is a critical point of 
\begin{equation*}
\h \mapsto - \I\left( \Phi _{1}(\h)\right) +\frac{1}{2}\left\Vert \h\right\Vert _{\CM}^{2}
\end{equation*}
so that for all $\h \in \CM$, 
\begin{equation} \label{equ:FOOnew}
 \left\langle \h^x,%
\h\right\rangle = \I^{\prime }\left( x\right) \left\langle D\Phi _{1}\left( \h%
^{x}\right) ,\h\right\rangle \ .
\end{equation}
(ii) with $g_1 (\omega) = G_1 (\WW) = \partial_\eps |_{\eps = 0 } \Phi _{1} (\h^x + \eps \W)$, we have
\begin{equation} \label{equ:relateStochIntwithIprime}
\int_0^1 \dot{\h}^{x}d\W = \I^{\prime }\left( x\right) g_{1} \ .
\end{equation}
\end{lemma}
\begin{proof}
Write $\Phi^{\h} \equiv \Phi(\h)$. For fixed $\h \in \CM$, define
\begin{equation*}
u\left( t\right) :=-\I\left( \Phi _{1}^{\h^x+t\h%
}\right) +\frac{1}{2}\left\Vert \h^x+t\h\right\Vert _{%
\CM}^{2}\geq 0
\end{equation*}%
with equality at $t=0$ (since $x=\Phi _{1}^{\h^x}$ and $%
\I\left( x\right) =\frac{1}{2}\left\Vert \h^x\right\Vert _{%
\CM}^{2}$) and non-negativity for all $t$ because $\h%
^{x}+t\h$ is an admissible control for reaching $\tilde{x}=\Phi
_{1}^{\h^x+t\h}$ (so that $\I\left( \tilde{x}\right)
=\inf \left\{ ...\right\} \leq \frac{1}{2}\left\Vert \h^x+t%
\h\right\Vert _{\CM}^{2}$.)\newline
By assumption, $\I$ is $C^1$ in a neighbourhood of $x$ and since $\Phi$ is also (Fr\'echet) $C^1$ it follows that $u=u(t)$ is $C^1$ near $t=0$. Since $u(0)=0$ and $u \ge 0$ we must have $\dot{u}\left( 0\right) =0$. 
In other words, $%
\h^x$ is a critical point for 
\begin{equation*}
\CM\ni \h\mapsto -\I\left( \Phi _{1}^{\h%
}\right) +\frac{1}{2}\left\Vert \h\right\Vert _{\CM}^{2}.
\end{equation*}%
The functional derivative of this map at $\h^x$ must
hence be zero. In particular, for all $\h \in \CM$, 
\begin{eqnarray*}
0 &\equiv &-\I^{\prime }\left( \Phi _{1}^{\h^x}\right)
\left\langle D\Phi _{1}\left( \h^x\right) ,\h%
\right\rangle +\left\langle \h^x,\h\right\rangle \\
&=&-\I^{\prime }\left( x\right) \left\langle D\Phi _{1}\left( \h%
^{x}\right) ,\h\right\rangle +\left\langle \h^x,%
\h\right\rangle .
\end{eqnarray*}%
{
Point (ii) is then a consequence of (i) and Lemma \ref{lem:gv}.
}

\end{proof}

\begin{remark}[Tangent space of admissible controls at $\h^x$] Combination of (\ref{equ:tan}) with (\ref{equ:FOOnew}) shows that
\begin{equation}\label{def:H0}
     T_{\h^{x}}\mathcal{K}^{x} = \{ \h^x \}^\perp =: \TH \, .
\end{equation}


\end{remark}

\subsection{Non-degeneracy} \label{app:SOO}

To this, recall that $\h^x$ was (by assumption) an energy minimizer which led to the first order 
optimality condition (\ref{equ:FOOnew}). As in calculus, being a minimizer tells us something about the
sign of the second derivative. To this end, consider the Hessian of \[I : \h \mapsto \frac{1}{2}\|\h\|_{\CM}^2\]
at $\h^x \in \mathcal{K}^{x} \subset \H$ as quadratic form on as a quadratic form on $H_0$. 
For every $\k \in \H_0 = T_{\h^{x}}\mathcal{K}^{x}$ we can find a $\mathcal{K}^{x}$-valued $\C^1$-curve $\varepsilon \mapsto \kappa^\varepsilon$, defined near $\varepsilon=0$, s.t.
$\kappa^0=\h^x$ and $\tfrac{d \kappa^\varepsilon}{d \varepsilon} _{| \eps=0} = \k$, and
$$ I''(\h^x)[\k,\k] = \frac{d^2}{d \varepsilon^2}_{|\epsilon=0} I(\kappa^\varepsilon)= \lim_{\varepsilon \to 0} \frac{1}{\varepsilon^2} \left(\|\kappa^\eps\|_{\CM}^2 -\|\h^x\|_{\CM}^2\right) \ge 0,$$
justified by criticality, i.e. $\frac{d}{d \varepsilon}_{|\epsilon=0} I(\kappa^\varepsilon)= \langle \h^x, \k \rangle = 0$, and then minimality of $\h^x$. 
For an explicit expression of $I'' (\h^x)$ we introduce $\A=\A^x$, a bilinear operator on $\H$ given by
$$\A[ \k , \l] = D^2\Phi_1(\h^x)\left[P_{\h^x}^\perp \k, P_{\h^x}^\perp \l\right],$$
where $P_{\h^x}^\perp$ is the projection on $\TH = \{ \h^x \}^\perp$. 
In other words, like in differential geometry, the Hessian $\A=\A^x$ is a quadratic form on the tangent-space $T_{\h^{x}}\mathcal{K}^{x} = \TH$.

We then have a more explicit description for  $I''(\h^x)$ in terms of $\A$, cf. assumption (A5c).

\begin{lemma} \label{lem:HessianExplicit} 
Let $\k \in \TH = T_{\h^x} \mathcal{K}^x$, then one has
$$I''(\h^x)[\k,\k] = \|\k\|^2 - q^x A[\k,\k]$$
where $q^x$ is given by Lemma \ref{lem:appendix-2pre}, and equals if $\Lambda'(x)$ when $\Lambda$ is $C^1$ at $x$. 
\end{lemma}

%

\begin{proof}
Fix $\k \in H'$ and $\kappa^\varepsilon \in  \mathcal{K}^x$ with 
$$\kappa^\varepsilon = \h^x + \eps \k +  r^\varepsilon$$
with $\|r^\eps\| = o(\eps)$. Then the constraint $\Phi(\kappa^\varepsilon) = x$ implies by Taylor expansion at order $2$ that
$$D\Phi(\h^x)[r^\varepsilon] + \eps^2 \frac{1}{2} D^2 \Phi_1(\h^x)[\k,\k]=o(\eps^2)$$
and by Lemma  \ref{lem:appendix-2pre}
$$ \eps^{-2} \left\langle \h^x, r^\eps \right\rangle = - \frac{q^x}{2}  D^2 \Phi_1(\h^x)[\k,\k] + o(1).$$
Hence we have that
$$I''(\h^x)[\k,\k] =\lim_{\eps \to 0} \eps^{-2} \left(\|\eps \k + r^\eps\|_{\H}^2  + 2 \left\langle \h^x,  \eps \k + r^\eps \right\rangle\right) = \|\k\|^2_{\H} - q^x D^2 \Phi_1(\h^x)[\k,\k].$$
\end{proof}


%
%
%
%
%

The following lemma not only provides conditions for the minimizer $\h^x$, and therefore the energy function $\Lambda(x)$, to be locally regular in $x$, but also shows that the relevant non-degeneracy conditions are ``open'' in the precise sense given below. (As was seen in the example of RoughVol, Corollary \ref{thm:main44}, this is very convenient in checking condition (A5) in a concrete situation.) 
%

\begin{lemma} \label{lem:hx} Assume that $\Phi_1: \H \to \R$ is weakly continuous on bounded sets.
Assume that $\h^x$ is the unique minimizer in $\mathcal{K}^x$, that $\Phi_1$ is $C^{k+1}$-Fr\'echet differentiable for some $k\ge1$ in a neighbourhood of $\h^x$, that $D \Phi(\h^x) \neq 0$, and that $\h^x$ is nondegenerate in the sense that $I''(\h^x)[\k,\k]>0$, for all $\k \ne 0$. Then there exists a neighborhood $V$ of $x$ and $C^{k}$ maps $\hat{h} : V \to \H$, $\hat{q} : V \to \R$ such that for all $y \in V$
\[ \hat{h}(y) = \h^y = \arg\min \{\| \h \|^2_{\H} , \; \; \Phi_1(\h) = y \} \]
is the unique minimizer in $\mathcal{K}^y$, such that $D \Phi(\hat{h}(y)) \neq 0$, $\hat{q}(y) = q^y $ is the (unique) associated Lagrange multiplier from Lemma \ref{lem:appendix-2pre}, and the minimizer $\hat{h}(y) = \h^y$ is again non-degenerate. Finally, the energy function $\Lambda$ is $C^{k+1}$ on $V$.
\end{lemma}

\begin{proof} We start with the remark that, for $y$ sufficiently close to $x$, we have $\mathcal{K}^{y} \ne \emptyset$. Indeed, $\langle D\Phi_1 (\h^x), \k \rangle \ne 0$ for some $\k \in \H$, which we may take of length one, the inverse function theorem applied to $\eps \mapsto \Phi_1 (\h^x + \eps k)$ at $\eps = 0$, then guarantees local controllability together with local boundedness of the energy function, $\Lambda(y) =\Lambda(\Phi_1 (\h^x + \eps (y) k)) \le \tfrac{1}{2} \| h^x + \eps (y) k \|^2_{\H} \lesssim \Lambda (x) + \eps(y)^2$. A $C^k$-map $\Psi : \H \times \mathbb{R} \rightarrow \H \times \mathbb{R}$ is defined, for fixed $x$ and noting $\Psi (\h^x, q^x) = (0, x)$, by
\[ \Psi : (\h, q) \mapsto (\h - q D \Phi_1 (\h), \Phi_1 (\h)). \]
One checks, as in \cite[Prop. 5.1]{kusuoka2008remark} that the
Fr\'echet derivative  $D \Psi (\h^x, q^x)$ of $D \Psi$ at $(\h^x, q^x)$ is
non-degenerate. The inverse function theorem then yields local invertibility of $\Psi$ and more precisely $C^k$-functions
$\hat{h}, \hat{q}$ such that for $y$ in neighbourhood of $x$,
\[ \Psi (\hat{h} (y), \hat{q} (y)) = (0, y), \quad \text{and} \quad (\hat{h}
   (x), \hat{q} (x) ) = (\h^x, q^x) . \]
 One sees moreover, exactly as in the proof of \cite[Prop. 5.1]{kusuoka2008remark}, that $
 \Lambda (y) = \tfrac{1}{2} \| \hat{h} (y) \|^2_\H$ which identifies $\hat{h} (y)$ as minimizer in $
 \mathcal{K}^y$, but uniqueness was not dealt with in \cite{kusuoka2008remark}. Let $\h^y$ be a (possibly 
 different) minimizer in $\mathcal{K}^y$. 
 Provided $D \Phi _{1} ({\h^y}) \ne 0$, Lemma \ref{lem:appendix-2pre} yields a Lagrange multiplier 
 $q^y$ such that, using further Lemma  \ref{lem:appendix-2}, we have $\Psi (\h^y,q^y) = (0,y)$, 
 which identifies $(\h^y,q^y)=(\hat{h} (y), \hat{q} (y))$ by local invertibility. To see that $D \Phi _{1} (
 {\h^y})$ is indeed non-zero, for $y$ close enough to $x$, assume to the contrary that $D \Phi _{1} 
 ({\h^{y_n}}) =0$ for some sequence $y_n \to x$ and minimizers $\h^{y_n} \in \mathcal{K}^{y_n}$. 
 
Since $\| \h^{y_n} \|^2_{\H} = 2 \Lambda (y^n)$  is bounded, by passing to a subsequence, we 
 may assume that $\h^{y_n} \to \bar \h$ weakly. In turn, weak continuity ensures that
 $\Phi(\h^{y_n}) \to \Phi(\bar{\h})$, so that $\bar{\h} \in \mathcal{K}^x$. Since the norm is weakly l.s.c. we have that 
\[ \|\bar{\h}\|_{\H} \leq \liminf _n \|\h^{y_n}\|_{\H} =  \liminf _n \|\hat{h}(y^n)\|_{\H} = \|\hat{h}(x)\|_{\H} =\|\h^x\|_{\H}, \] 
hence $\bar{\h} =h^x $, and the inequalities above are equalities. But this further implies that $\h^{y_n} \to \h^x$ strongly in $H$, hence $D \Phi _{1} (\h^{y_n}) \to D \Phi _{1} ({\h^x}) \ne 0$, 
contradiction.
We thus have shown that $\hat{h} (y) = \h^y$ is the unique minimizer in $\mathcal{K}^y$, (by IFT 
construction) strongly continuous from $y \in \R$ into $\H$. Using also continuity of $\hat{q}$ and, 
consequence of our Fr\'echet regularity assumption, continuity of $D^2 \Phi_1$, Lemma 
\ref{lem:HessianExplicit} shows that $I''(\h^y)$ is again non-degenerate. 
Finally, $C^k$-regularity of $\hat{h}$, $k \ge 1$, implies the same regularity for $\Lambda (y) = \frac{1}{2} \| \h^y \|^2_{\H}$. 
But then Lemma \ref{lem:appendix-2} shows that $\Lambda'  =  \hat{q} \in C^k$ so that in fact 
$\Lambda \in C^{k+1}$, always in some neighbourhood of $x$.
\end{proof}

\section{From option prices to implied volatilities}
%
%
%

\begin{proposition} Set $\eps = t^{1/2}, \bar \eps = t^{H}, k_\eps = x \eps^{1-2H} >0$ and consider a call price expansion of the form, with $\sigma_x^2 \equiv 2\Lambda(x) / (\Lambda'(x)^2)$,
\begin{equation} 
       c(\eps^2,k_\eps) \sim \exp \left( {  - \frac{\I(x)}{\beps^{2}}  } \right)  \eps \, \beps^2 
       \frac{A(x)} 
       { (\I'(x))^2  \sigma_x \sqrt{2\pi}}\ \ \ \text{as $\eps \downarrow0$} .
\end{equation}
Then the implied (squared) volatility expansion (\ref{e:IVI}) holds. The same expansion is implied by a put price expansion of the form
(\ref{equ:main0000}).
\end{proposition}

%

\begin{proof}Write $L_t = - \log c(t,k_t)$ and apply Gao--Lee  \cite[Corollary 7.1 - Equation (7.2)]{gao2014asymptotics}. Then
\begin{equation}\label{limit:gaolee2}
\Bigg|
\frac{1}{t}G^2_-(k_t,L_t-\frac{3}{2}\log(L_t)+\log(\frac{k_t}{4\sqrt{\pi}}))
-\sigma_I^2(k_t)
\Bigg|
=o\bigg(\frac{k_t^2}{L_t^2 t }\bigg)
\end{equation} 
where $G_-(k,u)$ denotes $\sqrt{2}(\sqrt{u+k}-\sqrt{u})$ and Gao-Lee $V$ is $\sqrt{t}\sigma_I(k_t)$.  We compute
\[
\begin{split}
&L_t-\frac{3}{2}\log(L_t)+\log(\frac{k_t}{4\sqrt{\pi}})\\
&=
\frac{\Lambda(x)}{t^{2H}}  - \log \frac{A(x)}{ (\Lambda'(x))^2  \sigma_x \sqrt{2\pi}}
-\frac{3}{2}\log(L_t)+\log(\frac{k_t}{4\sqrt{\pi}})
- (1/2+2H)\log t +o(1)
\end{split}
\]
and take care of the logarithmic terms in $t$:
\[
\begin{split}
&-\frac{3}{2}\log(L_t)+\log(\frac{k_t}{4\sqrt{\pi}})
- (1/2+2H)\log t\\
&= (- (1/2+2H)+3H+(1/2-H)) \log t  -\frac{3}{2}\log \Lambda(x)
+\log(\frac{x}{4\sqrt{\pi}}) +o(1)
\\
&= -\frac{3}{2}\log \Lambda(x)
+\log\big(\frac{x}{4\sqrt{\pi}}\big) +o(1) \;.
\end{split}
\]
So
\[
L_t-\frac{3}{2}\log(L_t)+\log(\frac{k_t}{4\sqrt{\pi}})-
\frac{\Lambda(x)}{t^{2H}}
\sim  - \log \frac{ A(x)(2\Lambda(x))^{3/2} }{ (\Lambda'(x))^2  \sigma_x x}  \;.
\]
Now, we note that
\begin{equation}\label{exp:G:gaolee}
G^2_-(k,u)= 
\frac{k^2}{2u}
-
\frac{k^3}{4u^2}
+
O\bigg(\frac{k^4}{u^3}\bigg)
\end{equation}
for $k/u \downarrow 0$. The second term in the right-hand side will be asymptotically equivalent to
\begin{equation}\label{second_term_gaolee}
-
\frac{k_t^3}{4 t L_t^2}
=
\begin{cases}
 o(t^{2H}) & \mbox{ if } H<1/2 \\ 
\frac{x^2}{2\Lambda(x)^2} \log(\exp( -
\frac{ x
}{ 2}))t+o(t)
 & \mbox{ if } H=1/2 \;.
\end{cases} 
\end{equation}
In both cases, the next term $\frac{k_t^4}{t L_t^3} = O(t^{1+2H})$ will be negligible.

\noindent If $H<1/2$, the equations above tell us that 
\[
\begin{split}
\frac{1}{t}G^2_-(k_t,L_t-\frac{3}{2}\log(L_t)+\log(\frac{k_t}{4\sqrt{\pi}}))
&=
\frac{x^2t^{1-2H}}{2t \big(\frac{\Lambda(x)}{t^{2H}}
- \log \frac{ A(x)(2\Lambda(x))^{3/2} }{ (\Lambda'(x))^2  \sigma_x x} +o(1)
\big) } +o(t^{2H})\\
&=
\frac{x^2}{2\Lambda(x)}
+
t^{2H}
\frac{x^2}{2\Lambda(x)^2}\log \frac{2 A(x) \Lambda(x)  }{ x \Lambda'(x)}+o(t^{2H}). 
\end{split}
\]
We apply (\ref{limit:gaolee2}), noticing that 
$o(k_t^2/(L_t^2 t))=o(t^{2H})$, and the claimed expansion follows for $H<1/2$.

In case $H=1/2$, the term \eqref{second_term_gaolee}  is also relevant. In this case the statement follows the same way, once we take this term into account in \eqref{exp:G:gaolee}.
\end{proof}

\bibliographystyle{plain}
\bibliography{roughvol}

\end{document}